\documentclass[12pt]{article}

\usepackage{amsmath} \usepackage{amsthm} \usepackage{amscd} 
\usepackage{amssymb} \usepackage{amsfonts} \usepackage{array}
\usepackage{graphics} \usepackage[all]{xy} \usepackage{color}
\usepackage{setspace} 
\usepackage{times}

\newtheorem{definition}{Definition}
\newtheorem{lemma}{Lemma}
\newtheorem{theorem}{Theorem}

\newtheorem{assumption}{Assumption}

\newcommand{\msf}[1]{{\mathsf{#1}}}
\newcommand{\mbf}[1]{{\mathbf{#1}}}
\newcommand{\mi}[1]{{\mathit{#1}}}
\newcommand{\mc}[1]{{\mathcal{#1}}}
\newcommand{\SUA}{\msf{S} \cup \msf{A}}
\newcommand{\STDUACUN}{\msf{STD} \cup \msf{ACUN}}

\newcommand{\wellty}[1]{\mathsf{well\mbox{-}typed}{#1}}
\newcommand{\std}{\mathsf{STD}}
\newcommand{\pk}{\mi{pk}}
\newcommand{\sh}{\mi{sh}}
\newcommand{\acun}{\mathsf{ACUN}}
\newcommand{\StdOps}{\mi{StdOps}}
\newcommand{\lra}{\leftrightarrow}
\newcommand{\Lra}{\Leftrightarrow}
\newcommand{\Llr}{\Longleftrightarrow}
\newcommand{\Ra}{\Rightarrow}
\newcommand{\EnCp}{\mi{EncSubt}}

\newcommand{\Vars}{\mi{Vars}}
\newcommand{\NewVars}{\mi{NewVars}}
\newcommand{\Constants}{\mi{Constants}}
\newcommand{\Agent}{\mi{Agent}}
\newcommand{\txor}{\texttt{XOR}}
\newcommand{\xor}{\mathtt{XOR}}
\newcommand{\xorseq}[3]{#1\oplus#2 \oplus #3}
\newcommand{\mxor}[2]{#1\oplus #2}

\newcommand{\Rules}{\mi{Rules}}

\newcommand{\Terms}{T(F,\Vars)}
\newcommand{\FuncTerms}{\mi{Terms}}
\newcommand{\type}{\mi{type}}

\newcommand{\appl}{\mathsf{applicable}}
\newcommand{\satisfiable}{\mathsf{satisfiable}}

\newcommand{\normalize}{\mi{normalize}}
\newcommand{\normal}{\msf{normal}}

\newcommand{\act}{\msf{active}}
\newcommand{\fa}{\forall}
\newcommand{\ex}{\exists}
\newcommand{\simple}{\msf{simple}}

\newcommand{\semibundle}{\msf{semi\mbox{-}bundle}}
\newcommand{\conseq}{\msf{conseq}}

\newcommand{\SubTerms}{\mi{SubTerms}}
\newcommand{\op}[1]{\mi{op_{#1}}}

\newcommand{\munut}{\mu\mbox{-}\msf{NUT}}
\newcommand{\munutsat}{\mu\mbox{-}\msf{NUT}\mbox{-}\mi{Satisfying}}
\newcommand{\nutsat}{\msf{NUT}\mbox{-}\mi{Satisfying}}
\newcommand{\tnut}{\textsf{NUT}}
\newcommand{\sfs}{\msf{secureForSecrecy}}
\newcommand{\lan}{\langle}
\newcommand{\ran}{\rangle}
\newcommand{\sqss}{\sqsubset}

\newcommand{\tacun}{\textsf{ACUN}}
\newcommand{\tstd}{\textsf{STD}}
\newcommand{\VarIdP}{\mi{VarIdP}}
\newcommand{\vip}{\mi{vip}}

\newcommand{\NewConstants}{\mi{NewConstants}}

\newcommand{\pure}{\msf{pure}}
\newcommand{\tuple}[2]{\lan #1, ~#2 \ran}

\newcommand{\inseq}[2]{#1 ~~\msf{in} ~~\msf{#2}}
\newcommand{\tif}{\text{if}}

\newcommand{\FreshVars}{\mi{FreshVars}}
\newcommand{\FreshCons}{\mi{FreshCons}}

\newcommand{\LTKeys}{\mi{LTKeys}}

\newcommand{\constraint}{\msf{constraint}}

\newcommand{\SecVars}{\mi{SecVars}}
\newcommand{\SecConstants}{\mi{SecConstants}}

\newcommand{\iik}{\mi{IIK}}

\newcommand{\true}{\msf{true}}

\newcommand{\nslxor}{\msf{NSL_{\oplus}}}

\newcommand{\Term}{\mi{Term}}
\newcommand{\penc}[2]{[#1]^{\to}_{#2}}
\newcommand{\senc}[2]{[#1]^{\lra}_{#2}}

\newcommand{\isocs}{\mi{isocs}}
\newcommand{\subiso}{\sigma_{\mi{iso}}}
\newcommand{\chisocs}{\mi{chisocs}}
\newcommand{\subcomb}{\sigma_{\mi{comb}}}
\newcommand{\combcs}{\mi{combcs}}
\newcommand{\chcombcs}{\mi{chcombcs}}
\newcommand{\childseq}{\msf{childseq}}

\newcommand{\Th}{\mi{Th}}
\newcommand{\XorTerms}{\mi{XorTerms}}
		
\begin{document}

\title{How to prevent type-flaw and multi-protocol attacks on cryptographic protocols under Exclusive-OR}

\author{Sreekanth Malladi\footnote{Dakota State University, Madison, SD - 57042.  Email: \texttt{Sreekanth.Malladi@dsu.edu}} }

\pagestyle{plain}
  
\maketitle

\begin{abstract}

  Type-flaw attacks and multi-protocol attacks are notorious threats to cryptographic protocol security. They are arguably the most commonly reported attacks on protocols. For nearly fifteen years, researchers have continuously emphasized the importance of preventing these attacks.
  
  In their classical works, Heather et al. and Guttman et al. proved  that these could be prevented by tagging encrypted messages with distinct constants, in a standard protocol model with a free message algebra~\cite{HLS03,TG00}.

  However, most ``real-world'' protocols such as \texttt{SSL 3.0} are designed with the Exclusive-OR (\txor) operator that possesses algebraic properties, breaking the free algebra assumption. These algebraic properties induce equational theories that need to be considered when analyzing protocols that use the operator.

This is the problem we consider in this paper: We prove that, under certain assumptions, tagging encrypted components \emph{still} prevents type-flaw and multi-protocol attacks even in the presence of the \txor{} 
operator and its algebraic properties.
	
\end{abstract}


\newpage

\tableofcontents
	
\newpage
\section{Introduction}

A {\em type-flaw attack} on a protocol is an attack where a message
variable of one type is essentially substituted with a message of a
different type, to cause a violation of a security
property. Preventing type-flaw attacks is crucial for security
protocols since they are frequently reported in
literature~\cite{CJ97,Lowe96,Mea03}.

In their pioneer work, Heather et al. proved that pairing constants
called ``tags'' with each message, prevents type-flaw
attacks~\cite{HLS03}. However, Heather et al.'s work considered a
basic protocol model with a free message algebra. Operators such as
Exclusive-OR possess algebraic properties that violate the free
algebra assumption, by inducing equational theories.

Another very important problem for security protocols is the problem
of multiple protocols executing in parallel. This was shown to be a
major cause for attacks on protocols~\cite{KSW97,Cremers06}. In an
outstanding work, Guttman et al. tackled this problem and proved that
if distinct protocol identifiers were to be inserted as tags inside
all encrypted components, multi-protocol attacks can be
prevented~\cite{TG00}, in the same year and conference as that of
Heather et al.'s paper~\cite{HLS00}. However, like Heather et al.,
they too consider a basic and standard model with a free
term algebra.

Recent focus in research projects world-wide has been to extend
protocol analysis to protocols that use operators possessing algebraic
properties, to accommodate ``real-world'' protocols such as \texttt{SSL
  3.0} (e.g.~\cite{KustersT08,EscobarMM07}). Naturally, a
corresponding study into type-flaw and multi-protocol attacks would be both crucial and
interesting.

These are the problems we consider in this paper: We provide formal
proofs to establish that suggestions similar to those made by Heather et al. and
Guttman et al (to tag messages), are sufficient to prevent all type-flaw and multi-protocol attacks on security protocols even under the {\sf ACUN}\footnote{Associativity, Commutativity, existence of Unity and Nilpotence.} algebraic properties of the Exclusive-OR (\txor) operator.

Our proof approach extends that used by us in~\cite{Mall04}, is
general, and could be extended to other operators with equations such
as {\sf Inverse} and {\sf Idempotence} in addition to {\sf ACUN}. We
give some intuitions for these in our conclusion.

\paragraph*{Significance of the results to protocol analysis and verification.} 

Preventing type-flaw and multi-protocol attacks is obviously
beneficial to protocol security. However, there are also significant
advantages to protocol analysis and verification:

\begin{itemize}
\item As Heather et al. pointed out, preventing type-flaw attacks also
  allows many unbounded verification approaches
  (e.g.~\cite{THG98,Cohen00,HS00}) to be meaningful, since they assume
  the absence of type-flaw attacks;
\item Similarly, preventing multi-protocol attacks ensures that it is
  sufficient to analyze protocols in isolation, which was found to be
  much less complicated than analyzing in multi-protocol
  environments~\cite{Meadows99,KSW97};
\item Furthermore, knowing that these attacks can be prevented in
  advance, reduces complexity of analysis and substantially saves the
  search space for automated tools;
\item Preventing type-flaw attacks is a crucial step in achieving
  decidability results for protocol security, as identified
  in~\cite{Lowe99,RS-FST03}. With decidability results in place, protocol
  verification can be reduced to a trivial problem of analyzing a single
  session of a protocol, to conclude its security.
\end{itemize}

The main concept used by our proofs is as follows. When terms
containing \txor{} are unified, the {\sf ACUN} theory does not affect
the unifier obtained, if all the terms that are \txor ed are tagged with
constants. Thus, unifiers for unification problems involving the standard 
operators and the \txor{} operator are obtained only using the algorithm
for the standard operators. Hence, the results that were possible for
the standard operators remain intact, even when the \txor{} operator is
used in constructing messages.

\paragraph*{Organization.} In Section~\ref{s.type-flaw-ex}, we will show 
that tagging can prevent type-flaw attacks under
\txor{} using an example. In Section~\ref{s.framework}, we will formalize our
framework to model protocols and their executions. In
Section~\ref{s.some-lemmas}, we will prove some useful lemmas. In Section~\ref{s.main-results}, we will use these lemmas to achieve our main results and conclude with a discussion of future and related works. We provide an index to the notation and terminology used in the paper in Appendix~\ref{ss.index} and a detailed description of Baader \& Schulz algorithm for combined theory unification~\cite{BS96} using an example in Appendix~\ref{ss.BSCA}.

\section{Tagging prevents type-flaw attacks under \txor{}: Example}\label{s.type-flaw-ex}

In this section, we show that tagging prevents type-flaw attacks under \txor{} on an example. The example helps in elucidating our proof strategy to achieve our main results in the subsequent sections.

Consider the adapted Needham-Schroeder-Lowe protocol
($\msf{NSL}_{\oplus}$) by Chevalier et al.~\cite{CKRT03}. We further
modify it by inserting numbers $1$, $2$, and $3$ inside each encrypted
message as suggested by Heather et al.~\cite{HLS03}:

\begin{center}
\begin{tabular}{|l|}
\hline

\parbox[3,5]{4in}{ 

  \vskip 0.2in

  \begin{center}
    \begin{tabular}{ll}

      {\bf Msg 1.} $A \to B : [1,N_A,A ]_{pk(B)}$ \\

      {\bf Msg 2.} $B \to A : [2,N_A \oplus B,N_B ]_{pk(A)}$\\

			{\bf Msg 3.} $A \to B : [3,N_B]_{pk(B)}$
\end{tabular}
\end{center}

} \\[0.6in]

\hline

\end{tabular}
\end{center}

\vskip 0.1in ($A$ and $B$ are agent variables; $N_A$, $N_B$ are nonce
variables; $[X]_Y$ represents $X$ encrypted with $Y$; $\pk(X)$ is the public-key of $X$).

A type-flaw attack is possible on this protocol even in the presence
of component numbering (originally presented in~\cite{MH08}):

\begin{center}
\begin{tabular}{|l|}
\hline

\parbox[3,5]{5.5in}{ 

\vskip 0.1in

\begin{center}
\begin{tabular}{ll}

  $\mbf{\textcolor{blue}{Msg}}$ $\alpha\mbf{.1.}$ $a \to i$ : $[1,n_a,a ]_{pk(i)}$ \\

  \textcolor{red}{Msg} $\beta$.1. $i(a) \to  b$ :  $[1,n_a \oplus b \oplus i,a ]_{pk(b)}$\\

  \textcolor{red}{Msg} $\beta$.2. $b \to  i(a)$ :  $[2,n_a \oplus b \oplus i \oplus b,n_b ]_{pk(a)}$\\

  $\mbf{\textcolor{blue}{Msg}}$ $\alpha\mbf{.2.}$ $i   \to a$ : $[2,n_a \oplus i,n_b]_{pk(a)}$ (replaying \textcolor{red}{Msg} $\beta$.2)\\

  $\mbf{\textcolor{blue}{Msg}}$ $\alpha\mbf{.3.}$ $a    \to i : [3,n_b ]_{pk(i)}$\\

  \textcolor{red}{Msg} $\beta$.3. $i(a) \to  b$ : $[3,n_b]_{pk(b)}$\\

\end{tabular}
\end{center} 

} \\[0.6in]

\hline

\end{tabular}
\end{center}

\vskip 0.1in

$i$ is the identity of the intruder or attacker; $i(x)$ denotes $i$ spoofing as $x$. We use lowercase now for agent identities and nonces ($a, b, n_a, n_b$), since this is a trace of the protocol execution, not the protocol specification.

Notice the type-flaw in the first message ($n_a\oplus b \oplus i$
substituted for the claimed $N_A$) that induces a type-flaw in the
second message as well. This is strictly a type-flaw attack, since 
without the type-flaw and consequently without exploiting the
algebraic properties, the attack is not possible.

The above attack can be avoided if tagging were to be adopted for the
elements of the \txor{} operator as well:

\begin{center}
\begin{tabular}{|l|}
\hline

\parbox[3,5]{4.5in}{ 

  \vskip 0.1in

  \begin{center}
    \begin{tabular}{ll}

      {\bf Msg 1.} $A \to B : [1, N_A, A]_{\mathit{pk}(B)}$ \\

      {\bf Msg 2.} $B \to A : [ 2, [4,N_A] \oplus [5,B], N_B ]_{\mathit{pk}(A)}$\\

{\bf Msg 3.} $A \to B : [3, N_B]_{pk(B)}$\\

\end{tabular}
\end{center}

} \\[0.4in]

\hline

\end{tabular}
\end{center}

\vskip 0.1in

Now {\bf \textcolor{red}{Msg}}  $\beta${\bf .2}  is not replayable as  {\bf \textcolor{blue}{Msg}}  $\alpha${\bf .2}  even when  $i(a)$  sends 
 {\bf \textcolor{red}{Msg}}  $\beta${\bf .1}  as 

\begin{center}  
{\bf \textcolor{red}{Msg}}  $\beta${\bf .1.} $i(a) \to b : [ 1, [4,n_a] \oplus [5,b] \oplus [5,i], a
]_{\mathit{pk}(b)}$,
\end{center}

since  {\bf \textcolor{red}{Msg}}  $\beta${\bf .2}  then becomes 

\begin{center}  
{\bf \textcolor{red}{Msg}}  $\beta${\bf .2.} $b \to i(a) : [ 2, [4, [4, n_a] \oplus [5, b] \oplus [5, i]]
\oplus [5, b], n_b]_{\mathit{pk}(a)}$.
\end{center} 

This is not replayable as the required: 

\[ \mbf{\textcolor{blue}{Msg}}~\alpha{\bf .2.}~i \to a
: [2, [4, n_a] \oplus [5, i], n_b]_{\mathit{pk}(a)}
\]

because, inside  {\bf \textcolor{red}{Msg}}  $\beta${\bf .2},  one
occurence of  $[5, b]$  is in  $[4, [4, n_a] \oplus [5, b] \oplus
[5, i]]$  and the other is outside. Hence, they cannot be canceled.

This concept can be best understood when we review the attack symbolically. The crux of the attack was the unification of terms, $N_A \oplus b$ (sent by agent $b$ inside Msg~2) with $n_a \oplus i$ (expected by agent $a$ inside Msg~2). The result is a substitution of $n_a \oplus i \oplus b$ with the type $\msf{nonce} \oplus \msf{agent} \oplus \msf{agent}$ to the nonce variable $N_A$, resulting in a mismatch of types.

When we prevented the attack by adding more tags, the terms $[4,N_A] \oplus [5,b]$ and $[4,n_a] \oplus [5,i]$ had to be unified. But they are not unifiable, since no substitution to the variable $N_A$ will make them equal under the \tacun{} theory for the $\oplus$ operator. 

Note that, a substitution of $n_a \oplus b \oplus i$ to $N_A$ will make them equal, if an additional equation, say $[w, x] \oplus [y, z] = [w\oplus y, x \oplus z]$ is considered in addition to the \tacun{} theory.

In this case,  $[4,N_A] \oplus [5,b]$ will become $[4 \oplus 5, n_a \oplus b \oplus i \oplus b]$, which is equal to $[4 \oplus 5, n_a \oplus i]$, which in turn is equal to the other term to be unified, $[4, n_a] \oplus [5,i]$.
However, in this paper, we consider only the \tacun{} algebraic properties of the $\oplus$ operator, but not equations where both the standard operators such as pairing and the \txor{} operator are combined. We do provide some insights into extending our work with such equations, in our conclusion.

In Sections~\ref{s.framework}, \ref{s.some-lemmas} and \ref{s.main-results}, we will prove formally that such tagging prevents all type-flaw and multi-protocol attacks on protocols in general, under the {\sf ACUN} theory.

\section{The Framework}\label{s.framework}

In this section, we will describe our formal framework to model the design and analysis of protocols, which we subsequently use to achieve the proofs for our main results in Section~\ref{s.main-results}.

We will define the term algebra in  Section~\ref{ss.term-algebra}, the protocol model in Section~\ref{ss.prot-model}, generating symbolic constraint sequences for protocol messages and checking their satisfiability in Section~\ref{ss.constraints}, the security properties desired of protocols and attacks on them in Section~\ref{ss.security-prop} and our main protocol design requirements to prevent type-flaw and multi-protocol attacks in Section~\ref{ss.NUT-MuNUT}.

\subsection{Term Algebra}\label{ss.term-algebra}

We will first introduce the construction of protocol messages using some basic elements and operators in Section~\ref{sss.terms}. We will then introduce equational unification in Section~\ref{sss.unification}.

We derive much of our concepts here from Tuengerthal's technical report~\cite{Tuengerthal-TR-2006} where he has provided an excellent and clear explanation of equational unification.

\subsubsection{Terms}\label{sss.terms}

We will use {\em italics} font for sets, functions, and operators. On the other hand, we will use {\sf sans-serif} font for predicates, equations and theories (described in Section~\ref{sss.unification}).

We denote the {\em term algebra} as $T(F,\Vars)$, where $\Vars$ is a set of variables, and $F$ is a set of function symbols or operators, called a {\em signature}. The terms in $T(F,\Vars)$ are called $F$-Terms. Further, 

\begin{itemize}
	\item $\Vars \subset T(F,\Vars)$;
	\item $(\fa f \in F)(\msf{arity}(f) > 0 \wedge t_1, \ldots, t_n \in T(F,\Vars)  \Ra f(t_1,\ldots,t_n) \in T(F,\Vars))$.
\end{itemize}

The set of nullary function symbols are called {\em constants}. We assume that every variable and constant have a ``type" such as $\Agent$, $\mi{Nonce}$ etc., returned by a function $\type()$. 

We define $F$ as $\StdOps \cup \{ \xor \} \cup \Constants$, where, 

\[	\StdOps = \{ \mi{sequence}, \mi{penc}, \mi{senc}, \pk, \sh \}.	\] 

Further, if $f \in F$ and $t_1,\ldots,t_n \in \Terms$ then, 

\[	\type(f(t_1,\ldots,t_n)) = f(\type(t_1),\ldots,\type(t_n)).	\]

$\mi{penc}$ and $\mi{senc}$ denote asymmetric and symmetric encryption operators respectively. $\pk$ and $\sh$ denote public-key and shared-key operators respectively. We assume that they will always be used  with one and two arguments respectively, that are of the type $\Agent$. 

We use some syntactic sugar in using some of these operators: 

\begin{eqnarray*}
  \mi{sequence}(t_1,\ldots,t_n)  &	= &	[t_1, \ldots, t_n], \\
						\mi{penc}(t,k)       &	= &  [t]^{\to}_k,   \\
						\mi{senc}(t,k) 			 &	= &  {[t]^{\lra}_k},   \\
				\xor(t_1,\ldots,t_n)		 &	=	&	t_1 \oplus \ldots \oplus t_n. 
\end{eqnarray*}

We will omit the superscripts $\lra$ and $\to$ for encryptions if the mode of encryption is contextually obvious or irrelevant.

We will write $\inseq{a}{[a_1,\ldots,a_n]}$ if $a \in \{a_1,\ldots,a_n\}$. We will denote the linear ordering relation of a sequence of elements, $s$, as $\prec_s$. For instance, if $s$ is a sequence such that $s = [s_1,\ldots,s_n]$, then, $(\fa i, j \in \{1,\ldots,n\})((i < j) \Ra (s_i \prec_s s_j))$.	




We define the subterm relation as follows: 

\begin{center}
$t \sqss t'$ iff $t' = f(t_1,\ldots,t_n)$ where $f \in 	F$ and $t \sqss t''$ for some $t'' \in \{ t_1, \ldots, t_n \}$.
\end{center}

We will use functions $\Vars()$, $\Constants()$, and $\SubTerms()$ on a single term or sets of terms, that return the variables, constants and subterms in them respectively. For instance, if $T$ is a set of terms,

\[	\SubTerms(T) = \{ t \mid (\ex t' \in T)(t \sqsubset t') \}.	\]

\subsubsection{Equational Unification}\label{sss.unification}

We will now introduce the concepts of unification under equational theories. We will start off with some basic definitions:

\begin{definition}{\em\bf [Substitution]}\label{d.sub}

A {\em substitution} is a tuple $\tuple{x}{X}$ (denoted $x/X$), where $x$ is a term and $X$ is a variable. Let $\sigma$ be a set of substitutions and $t$ be a term. Then,

\begin{center}
\begin{tabular}{lllll}
$t\sigma$ &   $=$ $t$, &  $\tif{}$ $t \in \mi{Constants}$, \\
     			&   $=$    $t'$, & $\tif{}$ $t'/t \in \sigma$, \\
      		&   $=$    $f(t_1\sigma,\ldots,t_n\sigma)$, & $\tif{}$ $f \in F, \text{and}$ $t = f(t_1,\ldots,t_n)$. 
\end{tabular}
\end{center}

\end{definition}

We extend this definition to define substitutions to a set of terms: If $T$ is a set of terms, then, $T \sigma = \{ t\sigma \mid t \in T \}$.

We will now introduce equational theories.

\begin{definition}{\em\bf [Identity and Equational Theory]}\label{d.equation-theory}

Given a signature $F$, and a set of variables $\Vars$, a set of {\em identities} $E$ is a subset of $T(F,\Vars) \times T(F,\Vars)$. We denote an identity as $t \cong t'$ where $t$ and $t'$ belong to $T(F,\Vars)$. An {\em equational theory} (or simply a theory) $=_E$ is the least congruence relation on $T(F,\Vars)$, that is closed under substitution and contains $E$. i.e., 

\[
	 =_E  ~:= 
	 ~\left\{ 
	 		R \mid 
	 			\begin{array}{c}
	 				R ~\mathrm{is ~a ~congruence~relation~on} ~T(F,\Vars), E \subseteq R, \mathrm{and} \\
	 				(\forall \sigma)(t \cong t' \in R \Ra t\sigma \cong t'\sigma \in R)
	 			\end{array}
	 \right\}					
\]

We write $t =_E t'$ if $(t,t')$ belongs to $=_E$.

\end{definition}

For the signature of this paper, we define two theories, $=_\std$ and $=_\acun$.

The theory $=_\std$ for $\StdOps$-Terms is based on a set of identities between syntactically equal terms, except for those made with the operator $\sh$:

\begin{center}
		$\std =$
			\begin{tabular}{rcl}
  		  $\{ [t_1,\ldots,t_n]$ &	$\cong$  & $[t_1,\ldots,t_n]$, \\
											$h(t)$  &  $\cong$  &  $h(t)$, \\
									$sig_k(t)$  &  $\cong$  &  $sig_k(t)$, \\
										$\pk(t)$  &  $\cong$  &  $\pk(t)$, \\
				 						$[t]_k$ 	&	 $\cong$	&  $[t]_k$, \\ 	
							$\sh(t_1,t_2)$  &  $\cong$  & $\sh(t_2,t_1)  \}$.	
			\end{tabular}
\end{center}

The theory $=_\acun$ is based on identities solely with the \txor{} ($\oplus$) operator, reflecting the algebraic properties of \txor: 
	
	\[ \acun = \{ t_1 \oplus (t_2 \oplus t_3) \cong (t_1 \oplus t_2) \oplus t_3 ,  \mxor{t_1}{t_2}  \cong  \mxor{t_2}{t_1}, \mxor{t}{0}  \cong  t,  \mxor{t}{t}  \cong 0 \}.
	\] 

We will say that a term $t$ is {\em pure} wrt theory $=_E$ iff there exists a substitution $\sigma$ and a term $t'$ such that $t = t'\sigma$ and either\footnote{Following Lowe~\cite{Lowe99}, we adopt functional programming convention and use an underscore (\_) in a formula, when the value in it doesn't affect the truthness of the formula.} $\_ \approx t'$ or $t' \approx \_$ belongs to $E$.

\[	\pure(t,=_E)   \Lra   (\ex t'; \sigma)((t = t'\sigma) \wedge ((t' \approx \_ \in E) \vee (\_ \approx t' \in E))).	\]

We will say that a term $t$ is an {\em alien subterm} of $t'$ wrt the theory $=_E$ iff it is not pure wrt $=_E$:

\[	\msf{ast}(t',t,=_E) \Lra  (t' \sqsubset t) \wedge \neg\pure(t',=_E).	\]

We will now describe equational unification. 

\begin{definition}{\em\bf [Unification Problem, Unifier]}\label{d.unification}

If $F$ is a signature and $E$ is a set of identities, then an $E$-{\em Unification Problem} over $F$ is a finite set of equations 

\[	\Gamma = \left\{ \begin{array}{c}  s_1 \stackrel{?}{=}_E t_1, \ldots, s_n \stackrel{?}{=}_E t_n \end{array} \right\} \]

between $F$-terms. A substitution $\sigma$ is called an $E$-{\em Unifier} for $\Gamma$ if $(\fa s \stackrel{?}{=}_E t \in \Gamma)(s\sigma =_E t\sigma)$. $U_E(\Gamma)$ is the set of all $E$-Unifiers of $\Gamma$. A $E$-Unification Problem is called $E$-{\em Unifiable} iff $U_E(\Gamma) \neq \{ \}$.

A {\em complete} set of $E$-Unifiers of an $E$-Unification Problem $\Gamma$ is a set $C$ of idempotent $E$-Unifiers of $\Gamma$ such that for each $\theta \in U_E(\Gamma)$ there exists $\sigma \in C$ with $\sigma \ge_E \theta$, where $\ge_E$ is a partial order on $U_E(\Gamma)$.

\end{definition}

An $E$-{\em Unification Algorithm} takes an $E$-Unification Problem $\Gamma$ and returns a finite, complete set of $E$-Unifiers.

Hence forth, we will abbreviate ``Unification Algorithm" to UA and ``Unification Problem" to UP.

Two theories $=_{E_1}$ and $=_{E_2}$ are {\em disjoint} if $E_1$ and $E_2$ do not use any common function symbols. UAs for two disjoint theories may be combined to output the complete set of unifiers for  UPs made with operators from both the theories, using Baader \& Schulz Combination Algorithm (BSCA)~\cite{BS96}. 


\subsection{Protocol Model}\label{ss.prot-model}

We will now introduce our protocol model, which is based on the strand space framework~\cite{THG98}. 

\begin{definition}{{\em\bf [Node, Strand, Protocol]}}\label{d.node-strand}
	A {\em node} is a tuple $\tuple{\pm}{t}$ denoted $\pm t$ where $t \in T(F,\Vars)$. 
  A {\em strand} is a sequence of nodes. 
  A {\em protocol} is a set of strands called ``roles".
\end{definition}

Informally, we write $+t$ if a node ``sends" term $t$ and $-t$ if it ``receives" $t$.  Further, if $\tuple{s}{t}$ is a node, then, $\tuple{s}{t}\sigma = \tuple{s}{t\sigma}$.

As an example for strands and protocols, consider the $\nslxor$ protocol presented in Section~\ref{s.type-flaw-ex}. This protocol that has two roles, $\mi{role_A}$ and $\mi{role_B}$. i.e., 

\[  \nslxor = \{ \mi{role_A}, \mi{role_B} \},  \]

where

\noindent
\begin{center}
\begin{tabular}{rll}
	$\mi{role_A}$ & $= [ +[1,A,N_A]_{\pk(B)}, -[2,N_A \oplus B,N_B]_{\pk(A)}, 						+[3,N_B]_{\pk(B)}  ]$, \text{and}\\ 
	$\mi{role_B}$ & $= [ -[1,A,N_A]_{\pk(B)}, +[2,N_A \oplus B,N_B]_{\pk(A)}, 						-[3,N_B]_{\pk(B)}  ]$.
\end{tabular}
\end{center}

We define a function $\FuncTerms()$ to return all the terms in the nodes of a strand. If $r$ is a strand, then,

\[	\FuncTerms(r)	=	\{	t \mid \tuple{\_}{t} ~\msf{in}	~r \}.	\]

We will also overload the functions $\Vars()$, $\Constants()$, and $\SubTerms()$ that were previously defined on sets of terms to strands in the obvious way. For instance, if $s$ is a strand, then,

\[
	\begin{array}{l}
		\SubTerms(s) = \{ t \mid (\ex t')((t' \in \FuncTerms(s)) \wedge (t \sqsubset t'))	\},	\\
		\Vars(s)	=  \Vars(\SubTerms(s)),	\\
		\Constants(s) = \Constants(\SubTerms(s)).
	\end{array}
\]

A {\em semi-bundle}	$S$ for a protocol $P$ is a set of strands formed by applying substitutions to some of the variables in the strands of $P$: If $P$ is a protocol, then, 

\[	
	\semibundle(S,P) \Ra  (\fa s \in S)((\ex r \in P; \sigma)( s = r\sigma )).
\]

For instance, $S = \{ s_{a1}, s_{a2}, s_{b1}, s_{b2} \}$ below is a semi-bundle for the $\nslxor$ protocol with two strands per role of the protocol: 

\begin{center}
\begin{tabular}{rcl}

$s_{a1}$  &  $=$  &  $[ +[a1,n_{a1}]_{\pk(B1)}, -[n_{a1} \oplus B1,N_{B1}]_{\pk(A1)}, +[N_{B1}]_{\pk(B1)} ]$, \\
$s_{a2}$  &  $=$  &  $[ +[a2,n_{a2}]_{\pk(B2)}, -[n_{a2} \oplus B2,N_{B2}]_{\pk(A2)}, +[N_{B2}]_{\pk(B2)} ] $, \\
$s_{b1}$  &  $=$  &  $[ -[A_3,N_{A3}]_{\pk(b1)}, +[N_{A3} \oplus b1,n_{b1}]_{\pk(A3)}, -[n_{b1}]_{\pk(b1)}  ] $,  \\
$s_{b2}$  &  $=$  &  $[ -[A_4,N_{A4}]_{\pk(b2)}, +[N_{A4} \oplus b2,n_{b2}]_{\pk(A4)}, -[n_{b2}]_{\pk(b2)}  ] $.

\end{tabular}
\end{center}

({\em Note}: As stated earlier,  we use lower-case symbols for constants and upper-case   for variables).

We will assume that every protocol has a set of variables that are considered ``fresh variables" (e.g. Nonces and Session-keys). If $P$ is a protocol, then, $\FreshVars(P)$ denotes the set of fresh variables in $P$. We will call the constants substituted to fresh variables of a protocol in its semi-bundles as ``fresh constants" and denote them as $\FreshCons(S)$. i.e., If $\semibundle(S,P)$, then,

\[
\FreshCons(S) = 
\left\{ 
\begin{array}{c}
 x \mid  
 \left(
 	\begin{array}{c}
	 	\ex r \in P; s \in S; \\
	 	\sigma; X
	\end{array}
 \right)
	\left(
		\begin{array}{c}
			(r\sigma = s) \wedge (X \in \FreshVars(P)) \wedge \\ 
			(x = X\sigma) \wedge (x \in \Constants)
		\end{array}
	\right)
\end{array}
\right\}.
\]

We assume that some fresh variables are ``secret variables" and denote them as $\SecVars(P)$. We define ``$\SecConstants()$" to return ``secret constants" that were used to instantiate secret variables of a protocol: If $\semibundle(S,P)$, then, 

\[
\SecConstants(S) = 
\left\{ 
\begin{array}{c}
 x \mid 
 \left(
 	\begin{array}{c}
	 	\ex r \in P; s \in S; \\
	 	\sigma; X
	\end{array}
 \right)
	\left(
		\begin{array}{c}
			(r\sigma = s) \wedge (X \in \SecVars(P)) \wedge \\ 
			(x = X\sigma) \wedge (x \in \Constants)
		\end{array}
	\right)
\end{array}
\right\}.
\]

For instance, $N_A$ and $N_B$ are secret variables in the $\msf{NSL}_{\oplus}$ protocol and $n_{a1}, n_{a2}, n_{b1}, n_{b2}$ are the secret constants for its semi-bundle above.

We will lift the functions $\Vars()$, $\Constants()$, $\SubTerms()$, and $\FuncTerms()$ that were previously defined on sets of terms and strands, to sets of strands. For instance, if $S$ is a set of strands, then,

\[ 
	\begin{array}{l}
		\SubTerms(S) = \{ t  \mid  (\ex x \in S)(t \in \SubTerms(x)) \},	\\
		\Constants(S)	=	\Constants(\SubTerms(S)),	\\
		\Vars(S)	=  \Vars(\SubTerms(S)), \\
		\FuncTerms(S) = \{	t \mid (\ex s \in S)(t \in \FuncTerms(s))	\}.	
	\end{array}
\]

We denote the long-term shared-keys of a protocol $P$ as $\LTKeys(P)$, where,

\noindent
\begin{eqnarray*}
\LTKeys(P) &  =  &  \{ x  \mid  (\ex A, B)( (x = \mi{sh}(A,B)) \wedge (x \in \SubTerms(P)) ) \}. 
\end{eqnarray*}

To achieve our main results, we need to make some assumptions. Most of our assumptions are reasonable, not too restrictive for protocol design and in fact, good design practices that improve security.

Before we start off with our first assumption, we will define a predicate $\wellty()$ on substitutions such that a substitution  is said to be well-typed, if the type of the variable is the same as that of the term it is substituted for: 

\[  (\fa t \in \Terms; X \in \Vars)((\wellty(t/X) \Llr (\type(t) = \type(X)))).  \]

We extend $\wellty()$ on sets of substitutions such that a set of substitutions is well-typed if all its elements are well-typed:

\[ (\fa \sigma)(\wellty(\sigma) \Llr (\fa t/X \in \sigma)(\wellty(t/X))). \]

We will now use this predicate to describe our first assumption which states  that the substitutions that are used on roles to form semi-strands, are always well-typed. This assumption is needed to achieve our result on type-flaw attacks.

\begin{assumption}{{\em\bf (Honest agent substitutions are always well-typed)}}\label{a.wellty}

If $\sigma$ is a set of substitutions that was used on a role to form a semi-strand, then $\sigma$ is well-typed:

	\[	(\fa \sigma)(\semibundle(S,\_) \wedge (\_\sigma \in S) \Ra \wellty(\sigma)).  \]
\end{assumption}

As noted in~\cite{CD-fmsd08}, for protocol composition or independence to hold, we first need an assumption that long-term shared-keys are never sent as part of the payload of messages in protocols, but only used as encryption keys. Obviously, this is a prudent and secure design principle.

Without this assumption, there could be multi-protocol attacks even when Guttman-Thayer suggestion of tagging encryptions is followed. For instance, consider the following protocols:

\begin{center}
\begin{tabular}{|c|c|}
	\hline	
	$\mbf{P_1}$ 	&	   $\mbf{P_2}$ \\  \hline
	1. $a \to s : \sh(a,s)$		&		  1. $a \to b : [1,n_a]_{\sh(a,s)}$ \\
	\hline
\end{tabular}
\end{center}

	Now the message in the second protocol could be decrypted with $\sh(a,s)$ and $n_a$ could be derived from it, when it is run with the first protocol.
	
	To formalize this assumption, we define a relation {\em interm} denoted $\Subset$ on terms such that, a term $t$ is an interm of $t'$ if it is a subterm of $t'$, but does not appear as an encryption key or inside a hash or a private-key signature. Formally, 

\begin{itemize}

	\item $t \Subset t'$ if $t = t'$, 
	\item $t \Subset [t_1,\ldots,t_n]$ if $(t \Subset t_1 \vee \ldots \vee t 								\Subset  t_n)$,
	\item $t \Subset [t']_k$ if $(t \Subset t')$,
	\item $t \Subset t_1 \oplus \ldots \oplus t_n$ if $(t \Subset t_1) \vee \ldots 					\vee (t \Subset t_n)$.

\end{itemize}

Notice that an interm is also a subterm, but a subterm is not necessarily an interm. For instance, $n_a$ is an interm and a subterm of $n_a \oplus [a]^{\to}_{n_b}$, while $n_b$ is a subterm, but not an interm.

Interms are useful in referring to the plain text of encryptions and in general, the ``payload" of messages. i.e., everything that can be ``read" by the recipient of a term. Contrast these with the keys of encrypted terms, which can only be confirmed by decrypting with the corresponding inverses, but cannot be read (unless included in the plain-text).

\begin{assumption}\label{a.LTKeys}
	If $P$ is a protocol, then, there is no term of $P$ with a long-term key as an interm: 
	
	\[ (\fa t \in \SubTerms(P))( (\nexists t' \Subset t)(t' \in \LTKeys(P)) ). \]
	
\end{assumption}

It turns out that this assumption is not sufficient. As noted by an anonymous reviewer of a workshop version of this paper~\cite{MallFCS10}, we also need another assumption that if a variable is used as a subterm of a key, then there should be no message in which that variable is sent in plain (since a long-term shared-key could be substituted to the variable as a way around the previous assumption).

Hence, we state our next assumption as follows:

\begin{assumption}\label{a.Key-Var}

	If $[t]_k$ is a subterm of a protocol, then no variable of $k$ is an interm of the protocol:
	
	\[	
		(\fa [t]_k \in \SubTerms(P))(\nexists X \Subset k; t' \in \SubTerms(P))
		\left(
			\begin{array}{c}
					(t' \neq k) \wedge \\
					(X \in \Vars) \wedge \\
					(X \Subset t')
			\end{array}
		\right).
	\]

\end{assumption}

Next, we will make some assumptions on the initial intruder knowledge. We will denote the set of terms known to the intruder before protocols are run, $\iik$. We will first formalize the assumption that he knows the public-keys of all the agents:

\begin{assumption}\label{a.IIK1}
	$(\fa x \in \Constants)(\pk(x) \in \iik)$.
\end{assumption}

In addition, we will also assume that the attacker knows the values of all the constants that were substituted by honest agents for all the non-fresh variables (e.g. agent identities $a, b$ etc.), when they form semi-strands:

\begin{assumption}\label{a.IIK2}

Let $P$ be a protocol. Then, 
	
	\[	(\fa x/X \in \sigma; r \in P)
	\left(
	\left(
	\begin{array}{c}
		 \semibundle(S,P) \wedge 
		 (r\sigma \in S) \wedge \\
		 (x \in \Constants)  \wedge 
		 (X \notin \FreshVars(P))
	\end{array}
	\right)
	\Ra
					(x \in \iik)
	\right).
	\]

\end{assumption}

Finally, we make another conventional assumption about protocols, namely that honest agents do not reuse fresh values such as nonces and session-keys:

\begin{assumption}\label{a.freshness}
	
		Let $S_1, S_2$ be two different semi-bundles. Then, 
			\[	\FreshCons(S_1) \cap \FreshCons(S_2) = \{ \}.	\]

\end{assumption}

\subsection{Constraints and Satisfiability}\label{ss.constraints}

In this section, we will formalize the concepts of generating symbolic constraints from node interleavings of semi-bundles and also the application of symbolic reduction rules to determine satisfiability of those constraints. These concepts are derived from the works of Millen-Shmatikov~\cite{MS01} and Chevalier~\cite{Chev04}, who later extended Millen-Shmatikov's model with the \txor{} operator. 

Formalizing constraint satisfiability allows us to rigorously model and reason about protocol executions and the security properties held within the executions: A satisfiable constraint sequence leads to a substitution when rules are applied on it and the substitution can be applied on protocols to generate protocol executions.

\begin{definition}{{\em\bf [Constraints, Constraint sequences]}}\label{d.constraints}

A {\em constraint} is a tuple $\tuple{m}{T}$ denoted $m : T$, where $m$ is a term called the {\em target} and $T$ is a set of terms, called the {\em term set}:

\[	\constraint(\tuple{m}{T}) \Ra (m \in \Terms) \wedge (T \in \mathcal{P}(\Terms)). \]

A {\em constraint sequence} is a sequence of constraints. A {\em constraint sequence} is from a semi-bundle if its targets and terms in term sets belong to strands in the semi-bundle. i.e., If $S$ is a semi-bundle, then, $cs$ is a constraint sequence of $S$, or

	\[	\conseq(cs,S)~\tif  \]

\begin{itemize}

	\item[(a)]	every target in $cs$ is from a `$-$' node of a strand in $S$:
	
\[	(\fa m : T~\msf{in}~cs)((\ex s \in S; n~\msf{in}~s)(n = -m)).  
\] 
		
	\item[(b)]	every term in every term set of $cs$ is from a `$+$' node of a strand in $S$:
	
		\[		(\fa m : T~\msf{in}~cs; t \in T)( (\ex s \in S; n~\msf{in}~s)(n = +t) ).		\]

\end{itemize}

\end{definition}	
	
A ``simple" constraint is a constraint whose target term is a variable. i.e., A constraint $m : T$ is simple if $m$ is a variable:

\[	\msf{simple}( m : T ) \Ra (m \in \mi{Vars}).   \]

A ``simple" constraint sequence is a sequence with all simple constraints. i.e., If $cs$ is a constraint sequence, then,

\[ \msf{simple}(cs)  \Ra  (\fa c ~\msf{in} ~cs)(\msf{simple}(c)).  \]

The ``active constraint" of a constraint sequence is the constraint in the sequence whose prior constraints are all simple constraints:

\[	
\act(c,cs) \Ra ( (c ~\msf{in} ~cs) \wedge (\fa c' ~\msf{in} ~cs)( (c' \prec_{cs} c) \Ra  \simple(c') ) ).
\]

We denote the sequence of constraints before the active constraint $c$ of a constraint sequence $cs$ as $cs_<$ and those after $c$ as $cs_>$. i.e., 

\[	cs = cs_<^{\frown}c^{\frown}cs_>. \]

\noindent
if $\act(c,cs)$ is {\sf true}, where $^{\frown}$ is the sequence concatenation operator.

Next, we define some symbolic reduction rules that can be applied on the active constraint of a constraint sequence. We name the set of all such rules as $\Rules$ where 

\[	\Rules = \{ \msf{un, ksub, join, split, senc, penc, sdec, pdec, hash, sig, xor_l, xor_r}  \}.  \]

Before defining the rules, we will explain a notation. If $c = m : T$ is a constraint and $\tau$ is a set of substitutions, then,

\[	c\tau = m\tau : T\tau. \]

In Table~\ref{t.rules}, we define $\mi{Rules}$, that can be applied on the active constraint of a constraint sequence.

\begin{table*}
\begin{center}
\begin{tabular}{|c|c|c||c|c|c|}
	\hline 

  $\msf{concat}$   &   $[t_1,\ldots,t_n]:T$   &   $t_1:T$,\ldots,$t_n:T$   &   $\msf{split}$   &   $t:T \cup [t_1,\ldots,t_n]$	&	 $t:T \cup t_1 \cup \ldots \cup t_n$ \\  \hline
 
  $\msf{penc}$   &   $[m]^{\to}_k:T$   &   $k:T,m:T$   &   $\msf{pdec}$   &   $m:[t]^{\to}_{\mi{pk}(\epsilon)}\cup T$	  &	 $m:t\cup T$ \\  \hline

  $\msf{senc}$   &   $[m]^{\leftrightarrow}_k:T$   &   $k:T,m:T$   &   $\msf{sdec}$   &   $m:[t]^{\leftrightarrow}_k\cup T$	&	 $k:T,m:T\cup\{t,k\}$ \\  \hline

  $\msf{sig}$  & $sig_k(t):T$  &  $t:T$ &  
  $\msf{hash}$ & $h(t):T$					&	 $t:T$ \\  \hline

  $\msf{xor_r}$   &   $m : T \cup$   &   $t_2 \oplus \ldots \oplus t_n : T,$  &   $\msf{xor_l}$   &    $t_1 \oplus \ldots \oplus t_n : T$	 & $t_2 \oplus \ldots \oplus t_n : T$,  \\  
 
 &	$t_1 \oplus \ldots \oplus t_n$	&		$m : T \cup t_1$ &		&		&	 $t_1 : T$ 	 \\ \hline

\end{tabular}
\end{center}

\caption{Set of reduction rules, {\em Rules}}\label{t.rules}

\end{table*}

The first column is the name of the rule, the second and third columns are the active constraints before and after the application of the rule.

We define a predicate $\appl()$ on each of these rules, that is true if the rule under consideration is applicable on the active constraint of the given constraint sequence. The predicate takes the name of the rule, the input sequence $cs$, the output sequence $cs'$, input substitution $\sigma$, output substitution $\sigma'$, and the theory $\Th$ considered as arguments. For instance, we define $\msf{xor_r}$ as follows:

\[
	\appl(\msf{xor_r},cs,cs',\sigma,\sigma',\Th)  
	\Lra
	(\ex m, T, t)
	\left(  
	\begin{array}{l}
	\act(m:T \cup t_1 \oplus \ldots \oplus t_n, cs) \wedge (\sigma' = \sigma) \wedge \\
	 (cs' = cs_< ^{\frown}[t_2 \oplus \ldots \oplus t_n:T, 
	 m:T \cup t_1]^{\frown} cs_>) 	 
	\end{array}
	 \right)
\]

Note that we did not use brackets $\{ \}$ for singleton sets, to avoid notational clutter. For instance, we write $m : T \cup t_1$, instead of $m : T \cup \{t_1\}$ since it is unambiguous. 

We left out two important rules in the table, $\msf{un}$ and $\msf{ksub}$, that are the only rules that change the attacker substitution through unification. We describe them next:

\[
	\appl(\msf{un},cs,cs',\sigma,\sigma',=_E)  
	\Lra
	(\ex m, T, t)
	\left(  
	\begin{array}{l}
	\act(m:T \cup t, cs) \wedge 
	 (cs' = cs_< \tau^{\frown}cs_>\tau) \wedge \\
	 (\sigma' = \sigma \cup \tau) \wedge (\tau \in U_E(\{m  \stackrel{?}{=}_E t\}))
	\end{array}
	 \right)
\]

\[
	\appl(\msf{ksub},cs,cs',\sigma,\sigma',=_E)  
	\Lra
	(\ex m, T, t)
	\left(  
	\begin{array}{l}
	\act(m : T \cup [t]^{\to}_k, cs) \wedge \\
	 (cs' = cs_< \tau^{\frown}[m\tau : T\tau \cup [t]^{\to}_k\tau]^{\frown} 
	  cs_>\tau) \wedge \\
	 (\sigma' = \sigma \cup \tau) \wedge (\tau \in U_E(\{k \stackrel{?}{=}_E \mi{pk}(\epsilon)\}))
	\end{array}
	 \right)
\]

({\em Note}: $\epsilon$ is a constant of type $\Agent$, representing the name of the attacker, always belonging to $\iik$).

We will say that a constraint sequence $cs'$ is a {\em child constraint sequence} of another sequence $cs$, if it can be obtained after applying some reduction rules on $cs$ in the theory $\Th$:  

\[
	\childseq(cs,cs',\Th) \Lra
	(\ex r_1, \ldots, r_n \in \Rules)
	\left(
		\begin{array}{l}
				\appl(r_1,cs,cs_1,\sigma,\sigma_1,\Th) \wedge \\
				\appl(r_2,cs_1,cs_2,\sigma_1,\sigma_2,\Th) \wedge \ldots \wedge \\
				\appl(r_n,cs_{n-1},cs',\sigma_{n-1},\sigma_n,\Th) 
		\end{array}								
	\right).
\]

%

We now define ``normal" constraint sequences, where the active constraint does not have sequences on the target or in the term set and has stand-alone variables in the term set (also recall that by definition, the target term of an active constraint is not a variable):

\[
	\normal(cs) \Lra
	\left(	
	 \begin{array}{c}
	 \act(m : T, cs)\wedge \\
	 (\nexists t_1,\ldots, t_n)( [t_1,\ldots,t_n] = m )\wedge  \\
	 ( (\fa t \in T)( (\nexists t_1,\ldots, t_n)( [t_1,\ldots,t_n] = t) ) \wedge \\
	 (\fa t \in T)(t \notin \mi{Vars}) )
	 \end{array}
	\right)
\]

Next, we will define a recursive function, $\normalize()$, that maps constraints to constraint sequences such that:

\begin{tabbing}
	$\normalize(m : T)$ \=  $=$ \= $[m : T]$,  if $\normal(m : T)$;\\
	\>  $=$	\>  $\normalize(t_1 : T)^{\frown}\ldots^{\frown}\normalize(t_n : T)$ if $m = [t_1,\ldots,t_n]$; \\
	\>  $=$ \>	$\normalize(m : T' \cup t_1 \cup \ldots \cup t_n)$ 
			if $T = T' \cup [t_1,\ldots,t_n]$.
\end{tabbing}

We will now overload this function to apply it on constraint sequences as well:

\begin{tabbing}
	$\normalize(cs)$ \= $=$  \=  $cs$, if $\normal(cs)$ \\
	\>  $=$ \> 	$cs_<^{\frown}\normalize(c)^{\frown}cs_>$, if $\act(c,cs)$.
\end{tabbing}

We define satisfiability of constraints as a predicate ``$\satisfiable$" which is true if there is a sequence of applicable rules which reduce a given normal constraint sequence $cs$ to a simple constraint sequence $cs_n$, in a theory $\Th$, resulting in a substitution $\sigma_n$: 

\begin{equation}\label{e.satisfiable}
\begin{array}{c}
\satisfiable(cs,\sigma_n,\Th) \Ra \\
	(\ex r_1,\ldots,r_n \in \Rules)
	\left(
	\begin{array}{l}
 		\appl(r_1,cs,cs_1,\{\}, \sigma_1, \Th) \wedge \\
  	\appl(r_2,cs'_1,cs_2,\sigma_1,\sigma_2, \Th) \wedge \ldots \wedge \\
    \appl(r_n,cs'_{n-1},cs_n,\sigma_{n-1},\sigma_n, \Th) \wedge \\
    \msf{simple}(cs_n) \wedge\\
   	(\fa i \in \{ 1,\ldots,n\})(cs'_i = \normalize(cs_i))
	\end{array}
\right).
\end{array}
\end{equation}

Notice the last clause which requires that every constraint sequence be normalized before any rule is applied, when checking for satisfiability.

This definition of satisfiability may seem unusual, especially for the puritans,  since satisfiability is usually defined using attacker capabilities as operators on sets of ground terms to generate each target on constraints.

However, it was proven in \cite{Chev04} that the decision procedure on which our definition is based, is sound and complete with respect to attacker capabilities on ground terms in the presence of the algebraic properties of \txor. Hence, we defined it directly in terms of the decision procedure, since we will be using only that to prove our main theorem. We refer the interested reader to \cite{MS01} and \cite{Chev04} for more details on the underlying attacker operators, whose usage is equated to the decision procedure that we have used.

Note also that our definition only captures completeness of the decision procedure wrt satisfiability, not soundness, since that is the only aspect we need for our proofs in this paper.

\subsection{Security properties and attacks}\label{ss.security-prop}

Every security protocol is designed to achieve certain goals (e.g. secure key establishment, authentication). Correspondingly, every execution of a protocol is expected to satisfy some security properties. For instance, a key establishment protocol should not leak the key being established, which would be a violation of secrecy. Similarly, a key establishment protocol should not lead an honest agent to exchange a key with an attacker which would be a violation of both secrecy and authentication.

Security properties such as secrecy can be tested if they hold on executions of protocols, by forming semi-bundles of the protocols, forming constraint sequences from the semi-bundles, adding the desired property to be tested to the constraint sequences and then checking if the resulting constraint sequence is satisfiable.

For instance, consider the following constraint sequence from a semi-bundle of the $\nslxor$ protocol:

\begin{center}
\begin{tabular}{rcl}
$[1,N_A,A]_{\mi{pk}(b)}$ & : & $[1,n_a,a]_{\mi{pk}(B)} \cup \iik$ ($= T_1$) \\
$[2,n_a\oplus B,N_B ]_{\mi{pk}(a)}$ 
 & : & $[2,N_A \oplus b, n_b ]_{\mi{pk}(A)} \cup T_1$ ($= T_2$)  \\
$[3,n_b]_{\mi{pk}(b)}$ & : & $[3,N_B]_{\mi{pk}(B)} \cup T_2$ \\
$n_b$		&	:	&	$T_2$.
\end{tabular}
\end{center}

The first three constraints are obtained from a semi-bundle with one strand per role of the $\nslxor$ protocol. The last constraint is an artificial constraint added to them, to test if secrecy is violated in the sequence.

If the constraint sequence is solved by applying the rules previously defined, it shows that the nonce $n_b$, which is supposedly secret, can be obtained by the attacker by interleaving the messages of honest agents $a$ and $b$. Specifically, we would apply $\msf{penc}$ to the first constraint, and split it into $[1,N_A,A] : T_1$ and $\pk(b) : T_1$. We would then apply $\msf{pair}$ to split the former into three constraints: $1 : T_1, N_A : T_1$, and $A : T_1$. Next, rule $\msf{un}$ is applied on the second constraint, unifying terms $[2,n_a\oplus B,N_B ]_{\mi{pk}(a)}$ and $[2,N_A \oplus b, n_b ]_{\mi{pk}(A)}$. The resulting unifier $\{ n_a \oplus b \oplus i/N_A, \epsilon/B, n_b/N_B \}$,  is applied on the term in the third constraint, $[3,N_B]_{\mi{pk}(B)}$, making it $[3,n_b]_{\mi{pk}(\epsilon)}$. Finally, $n_b$ can be extracted from this term using $\msf{pdec}$ and $\msf{pair}$, satisfying the last constraint.

Our definition of type-flaw attacks is general, and is valid for any property such as secrecy that can be tested on satisfiable constraint sequences from semi-bundles of protocols.

\begin{definition}{{\em\bf [Type-flaw attack]}}\label{d.type-flaw-attack}

A protocol has a {\em type-flaw attack} in the theory $\Th$ iff there exists a semi-bundle from the protocol that has a constraint sequence satisfiable only with a substitution that is not well-typed: i.e., if $P$ is a protocol, then: 

\begin{flushleft}
	$(\fa \msf{cs}, S)
  \left(
\begin{array}{l}
  \semibundle(S,P)   \wedge \conseq(\msf{cs},S) \wedge \\
   (\ex \sigma)(\satisfiable(\msf{cs},\sigma, \Th) \wedge  \neg\wellty(\sigma)) \wedge \\
   (\nexists \sigma')( \satisfiable(\msf{cs},\sigma', \Th) \wedge \wellty(\sigma') )
\end{array}
	\right) \Lra \msf{typeFlawAttack}(P, \Th)$.
\end{flushleft}

\end{definition}

While our result on type-flaw attack is general and valid for any trace property, we achieve our other result on multi-protocol attacks in the context of secrecy (extensible to other properties such as authentication).  Accordingly, we provide a definition for the property below.

\begin{definition}{{\em\bf [Secrecy]}}\label{d.secrecy}

A protocol is {\em secure for secrecy} in the theory $\Th$, if no constraint sequence from semi-bundles of the protocol is satisfiable, after a constraint  with its target as a secret constant of the semi-bundle and its term set as the term set of the last constraint of the sequence is added as the last constraint of the sequence. i.e., if $P$ is a protocol, then, 
\[
\sfs(P,\Th) \Lra
 (\nexists \mi{sec}, \msf{cs}, S)
	\left(
	\begin{array}{c}
	  \semibundle(S,P) \wedge \conseq(\msf{cs},S) \wedge \\
   (\msf{cs} = [\_ : \_, \ldots, \_ : T]) \wedge \\
			   (\mi{sec} \in \SecConstants(S)) \wedge \\
			   \satisfiable(\msf{cs}^{\frown}[\mi{sec} : T], \sigma, \Th) 
	\end{array}
	\right).
\]

\end{definition}

\subsection{Main Requirements -- {\sf NUT} and $\munut$}\label{ss.NUT-MuNUT}

We now formulate our main requirements on protocol messages to prevent all type-flaw and multi-protocol attacks in the $=_{\SUA}$ theory\footnote{$\SUA$ is an abbreviation for $\STDUACUN$.}. The requirements are slight variations of the suggestions by Heather et al. and Guttman et al., who suggest inserting distinct component numbers inside encryptions. In a symbolic model, such component numbering guarantees \tnut{} (Non-Unifiability of encrypted Terms).

We will first define a function $\EnCp()$ that returns all the encrypted subterms of a term\footnote{$\mathcal{P}(X)$ is the power-set of the set $X$.}:

\[	\EnCp	:	\Terms	\to	\mathcal{P}(\Terms)	\]

where, if $m$ is a term, then, $\EnCp(m)$ is the set of all terms such that if $t$ belongs to the set, then $t$ must be a subterm of $m$ and is an encryption, hash or signature:

\[
	\EnCp(m) = 
	\left\{
	\begin{array}{c}
		 t   \mid   
		 (\ex t',k' \in \Terms)
		 \left(
		 \begin{array}{c}
	      (t \sqss m)	\wedge \\
		 ((t = [t']^{\to}_{k'}) \vee (t = 											[t']^{\lra}_{k'}) \vee  \\ 
		 (t = h(t')) \vee (t = \mi{sig}_{k'}(t'))) \\
	  	\end{array}  
	  \right)  
	\end{array}	   
	\right\}.
\]

Further, if $S$ is a set of strands, then, it's encrypted subterms are the encryptions of it's subterms:

\[	\EnCp(S) = \{ t \mid (\ex t')((t' \in \SubTerms(S)) \wedge (t \in \EnCp(t'))) \}.	\]

\begin{definition}{{\em\bf [NUT]}}\label{d.NUT}

A protocol $P$ is $\nutsat$, i.e.,

\[	\nutsat(P)~\text{iff}   \]

\begin{itemize}

	\item[(a)] An encrypted subterm of the protocol is not $\std$-Unifiable with any other non-variable subterm of the protocol:
	
		\[	
				(\fa t_1, t_2)
				\left(
					\left(
						\begin{array}{c}
							(t_2 \notin \Vars) \wedge \\
							(t_1 \in \EnCp(P))	\wedge \\
							(t_2 \in \SubTerms(P))	\wedge \\
							(t_1 \neq t_2)
						\end{array}
					\right)
				\Ra
						((\fa \sigma_1, \sigma_2)(U_{\std}(t_1\sigma_1, t_2\sigma_2) = \{\}))
				\right).
		\]

	\item[(b)] A key used in an asymmetric encryption is not a free variable:
	
	\[	(\fa t \in \EnCp(P))
			\left(
				\begin{array}{c}
					(\ex t', k)((t = [t']^{\to}_{k}) \Ra (k \notin \Vars))
				\end{array}
			\right).
	\]

	\item[(c)] If an \txor{} term, say $t_1 \oplus \ldots \oplus t_n$, is a  										subterm of $P$, then, no two terms in $\{ t_1, \ldots, t_n \}$ 									are $\std$-Unifiable, unless they are equal:
	
	\[
		(\fa t_1 \oplus \ldots \oplus t_n \in \SubTerms(P);
			 t, t')
			\left(
			 \begin{array}{c}
			 		(t, t' \in \{ t_1, \ldots, t_n \}) \wedge (t \neq t') \Ra \\
			 		(\fa \sigma, \sigma')(U_{\std}(t\sigma, t'\sigma') = \{\}) 
			 \end{array}
		\right).
	\]

\end{itemize}

\end{definition}

The first requirement can be satisfied by simply inserting distinct component numbers inside distinct encrypted subterms of a protocol, as was done in the $\nslxor$ protocol in Section \ref{s.type-flaw-ex}.

The second requirement can be satisfied by adding a distinct constant to the key of an asymmetric encryption, if it was a free variable. For instance, $[1, N_A, B]^{\to}_K$ can be transformed into $[1, N_A, B]^{\to}_{[2,K]}$.

The third requirement can also be satisfied in much the same way as the other two. We can add a distinct constant to each textually distinct variable inside an \txor{} term. For instance, the second message in the original $\nslxor$ protocol was
			
			\[	[2,N_A \oplus B,N_B ]_{\pk(A)}.	\]

With the number `2' inside this message and numbers `1' and `3' inside the others, the protocol satisfied the first requirement above, but was still vulnerable to an attack. The third requirement above requires that the second message be changed to,

			\[	[2,[4, N_A] \oplus [5, B],N_B ]_{\pk(A)},	\]

that prevents the attack.

Next we deal with multi-protocol environments. Our requirement defined below, namely $\munut$, ensures that encrypted terms in different protocols cannot be replayed into one another. The requirement is an extension of Guttman-Thayer's suggestion to make encrypted terms distinguishable across protocols, to include \txor{} as well.

We first define a set $\XorTerms$ as:

\[
	\{ t \mid (\ex t_1,\ldots,t_n \in T(F,\Vars))(t_1 \oplus \ldots \oplus t_n = t)	\}.
\]

We are now ready to state the main requirement formally:

\begin{definition}{{\em\bf [$\munut$]}}\label{d.munut}

	Two protocols $P_1$ and $P_2$ are $\munutsat$, i.e., $\munutsat(P_1,P_2)$ iff:

\begin{enumerate}

\item Encrypted subterms in both protocols are not $\std$-Unifiable after applying any substitutions to them: 

\[
		\begin{array}{c}
			(\fa t_1 \in \EnCp(P_1), t_2 \in \EnCp(P_2))((\fa  \sigma_1,\sigma_2)(U_{\std}(t_1\sigma_1, t_2\sigma_2)) = \{\}).
		\end{array}
\]
	
\item Subterms of \txor-terms of one protocol (that are not \txor-terms themselves), are not $\std$-Unifiable with any subterms of \txor-terms of the other protocol (that are not \txor-terms as well):

\[
		\left(
				\begin{array}{l}
					\fa t_1 \oplus \ldots \oplus t_n \in \SubTerms(P_1), \\
					 t'_1 \oplus \ldots \oplus t'_n \in \SubTerms(P_2); t, t'
				\end{array}
		\right)
				\left(
				\begin{array}{c}
						(t \in \{ t_1, \ldots, t_n\}) \wedge 
						(t' \in \{t'_1, \ldots, t'_n\})  \\
						(t_1, \ldots, t_n, t'_1, \ldots, t'_n \notin \XorTerms) \wedge \\
						\Ra (\fa \sigma, \sigma')(U_{\std}(t\sigma, t'\sigma') = \{\})
				\end{array}
				\right).
\]

\end{enumerate}

\end{definition}

The first requirement is the same as Guttman-Thayer suggestion. The second requirement extends it to the case of \txor-terms, which is our stated extension in this paper.

The $\msf{NSL}_{\oplus}$ protocol can be transformed to suit this requirement by tagging its encrypted messages as follows:

\begin{center}
\begin{tabular}{ll}

{\bf Msg 1.} $A \to B : [\msf{nsl}_{\oplus}, N_A, A]_{\mathit{pk}(B)}$ \\

{\bf Msg 2.} $B \to A : [\msf{nsl}_{\oplus}, [\msf{nsl}_{\oplus},N_A] \oplus [\msf{nsl}_{\oplus},B], N_B ]_{\mathit{pk}(A)}$ \\

{\bf Msg 3.} $A \to B : [\msf{nsl}_{\oplus}, N_B]_{pk(B)}$\\

\end{tabular}
\end{center}

The constant ``$\msf{nsl}_{\oplus}$" inside the encryptions can be encoded using some suitable bit-encoding when the protocol is implemented. Obviously, other protocols must have their encrypted subterms start with the names of those protocols.

We will later use this requirement in Section \ref{ss.multi-prot-result} to prove that this is sufficient to prevent all multi-protocol attacks on security protocols, even when they use the \txor{}  operator.

\section{Some Lemmas}\label{s.some-lemmas}

In this section, we provide some useful lemmas that we will use later in our main theorems. 

\begin{itemize}

\item In Section~\ref{ss.stdwellty}, we prove that if two non-variable $\StdOps$-terms were obtained by applying two well-typed substitutions for the same term, then the unifier for the two terms is necessarily well-typed;

\item In Section~\ref{ss.combined-unifier}, we first introduce Baader \& Schulz Combination Algorithm (BSCA) to find unifiers for UPs from two disjoint theories, say $=_{E_1}$ and $=_{E_2}$~\cite{BS96}. We will then prove that if the unifier for the $E_1$-UP from a given $(E_1 \cup E_2)$-UP, say $\Gamma$, is empty, then the combined unifier is simply equal to the unifier for the $E_2$-UP from $\Gamma$;

\item In Section~\ref{ss.ACUN-Probs-Only-Constants}, we prove that all \tacun-UPs formed by using BSCA on an original ($\SUA$)-UP that does not have free variables in \txor{} terms, have only constants as subterms.

\end{itemize}

\subsection{Well-typed standard terms unify only under well-typed unifiers}\label{ss.stdwellty}


In our first lemma, we prove that two $\StdOps$-terms obtained by instantiating the same $\StdOps$-term, with well-typed substitutions, unify only under a well-typed substitution:

\begin{lemma}  {\em\bf [Well-typed $\StdOps$-terms unify only under well-typed unifiers]} \label{l.stdwellty}

	If   $t$   is a non-variable term that is pure wrt $=_\std$ theory:
	
	\[	 (t \notin \Vars) \wedge \pure(t,=_\std),	\]
	
	and $t_1, t_2$ are two terms that are also pure wrt $=_\std$ theory,  and obtained by applying sets of substitutions   $\sigma_1$   and   $\sigma_2$   such that,
	
	\[	t_1 = t\sigma_1   ~\text{and}   ~t_2 = t\sigma_2,  \]
	
	and   $\sigma_1$,  $\sigma_2$   are well-typed:
	
	\[	\wellty(\sigma_1)  \wedge  \wellty(\sigma_2),  \]

and every $x/X \in \sigma_1 \cup \sigma_2$ is such that $x$ is pure wrt $=_\std$:

\[	(\fa x/X \in \sigma_1 \cup \sigma_2)(\pure(x,=_\std)),	\]

then, any unifier for   $t_1$   and   $t_2$, will be necessarily well-typed:
	
	\[	(\fa \tau)( (t_1\tau =_{\std} t_2\tau) \Ra \wellty(\tau) ).  \]

\end{lemma}

\begin{proof}

Let   $t = \mi{op}(t'_1,\ldots,t'_n)$   where   $\mi{op} \in \StdOps$.

Now, 

\begin{center}
\begin{tabular}{rclr}
	$t_1$  &  $=$  &  $t \sigma_1$   &      (from hypothesis), \\
	     &	=  &	$\mi{op}(t'_1\sigma_1,\ldots,t'_n\sigma_1)$  &     (from Def. \ref{d.sub})
\end{tabular}
\end{center}

Similarly,   $t_2 = \op{}(t'_1\sigma_2, \ldots, t'_n\sigma_2)$.  Let $\tau$ be a set of substitutions. Then, we have that,

\[	(t_1\tau =_{\std} t_2\tau)	\Lra (\fa i \in \{ 1, \ldots, n \})(t'_i\sigma_1\tau =_{\std} t'_i\sigma_2\tau).  \]

Without loss of generality, consider

\[	t'_1 \sigma_1\tau =_{\std} t'_1\sigma_2\tau.  \]

Then, since   $\sigma_1$   and   $\sigma_2$   are well-typed, will be well-typed when:

\begin{itemize}

\item	Both   $t'_1\sigma_1$   and   $t'_1\sigma_2$   are variables; or
\item $t'_1\sigma_1$   is a variable and   $t'_1\sigma_2$   is a constant; or
\item $t'_1\sigma_1$   is a constant and   $t'_1\sigma_2$   is a variable.

\end{itemize}

For instance, if   $(t'_1\sigma_1 \in \Vars)$   s.t.   $t'_1\sigma_1 = X$   and   $(t'_1\sigma_2 \in \Constants)$   s.t.   $t'_1\sigma_2 = y$,   then, since   $\wellty(\sigma_1)$   and   $\wellty(\sigma_2)$, we have,

\[	\mi{type}(t'_1) = \mi{type}(X) = \mi{type}(y).  \]

and $\wellty(y/X)$.

Thus, we conclude:

\begin{equation}\label{e.Vars-Constants-WellTy}
	\left(
	\begin{array}{c}
		((t'_1\sigma_1, t'_1\sigma_2 \in \Vars)  \vee \\
		(t'_1\sigma_1 \in \Constants; t'_1\sigma_2 \in \Vars)  \vee \\
		(t'_1\sigma_1 \in \Vars; t'_1\sigma_2 \in \Constants))  	
	\end{array}	
	\right)
	 \wedge  (t'_1\sigma_1\tau =_{\std} t'_1\sigma_2\tau)  \Ra  \wellty(\tau)
\end{equation}

Given this, let us now assume for the purpose of induction that a unifier for   $t'_1\sigma_1$   and   $t'_1\sigma_2$   will be well-typed when both   $t'_1\sigma_1$   and   $t'_1\sigma_2$   are compound terms. i.e.,

\begin{equation}\label{e.Compound-WellTy}
	(t'_1\sigma_1,t'_2\sigma_2 \notin \Vars \cup \Constants) \wedge
  (t'_1\sigma_1\tau =_{\std} t'_1\sigma_2\tau) \Ra \wellty(\tau).
\end{equation}

Combining (\ref{e.Vars-Constants-WellTy}) and (\ref{e.Compound-WellTy}), we can conclude that all the unifiers for  $t'_i\sigma_1$   and   $t'_i\sigma_2$   ($i \in \{1,\ldots,n\}$) are well-typed:

\[	(\fa i \in \{1,\ldots,n\})( (t'_i\sigma_1\tau =_{\std} t'_i\sigma_2\tau) \Ra \wellty(\tau) ).   \]

This implies that our hypothesis is true:

\[	(t\sigma_1\tau =_\std t\sigma_2\tau) \Ra \wellty(\tau).  \]

\end{proof}


\subsection{Combined unifier when one of the unifier is empty}\label{ss.combined-unifier}

Our next two lemmas are related to the combined unification of  $(E_1 \cup E_2)$-UPs, where $=_{E_1}$ and $=_{E_2}$ are disjoint.

We first define the variables of a UP, $\Gamma$, as $\Vars(\Gamma)$, where every element of $\Vars(\Gamma)$ is a variable and a subterm of a UP in $\Gamma$:

\[  \Vars(\Gamma) = \{ X \mid (\ex s \stackrel{?}{=} t \in \Gamma)(((X \sqsubset s) \vee (X \sqsubset t)) \wedge (X \in \Vars)) \}.  \]

Similarly, 

\[  \Constants(\Gamma) = \{ X \mid (\ex s \stackrel{?}{=} t \in \Gamma)(((X \sqsubset s) \vee (X \sqsubset t)) \wedge (X \in \Constants)) \}.  \]

Further, we will say that term $t$ belongs to a UP, say $\Gamma$, even if $t$ is one of the terms of one of the problems in $\Gamma$. i.e., 

\[
	(t \in \Gamma) \Lra (\ex t')(t \stackrel{?}{=} t' \in \Gamma).
\]

We will now explain how two UAs $A_{E_1}$ and $A_{E_2}$ for two disjoint theories $=_{E_1}$ and $=_{E_2}$ respectively, may be combined to output the unifiers for a $(E_1 \cup E_2)$-UP using Baader \& Schulz Combination Algorithm (BSCA)~\cite{BS96}.  We give a more detailed explanation in Appendix~\ref{ss.BSCA} using an example UP for the interested reader.

BSCA takes as input a $(E_1 \cup E_2)$-UP, say $\Gamma$, and applies some transformations on them to derive $\Gamma_{5.1}$ and $\Gamma_{5.2}$ that are sets of $E_1$-UP and $E_2$-UPs respectively. We outline the steps in this process below (we formalize these steps directly in Lemma~\ref{l.acun-nosub} where we use BSCA in detail):

\paragraph*{Step 1 (Purify terms)} BSCA first ``purifies" the given $(E = E_1 \cup E_2)$-UP, $\Gamma$, into a new UP, $\Gamma_1$, with the introduction of some new variables, such that, all the terms are ``pure" wrt $=_{E_1}$ or $=_{E_2}$.

\paragraph*{Step 2. (Purify problems)} Next, BSCA purifies $\Gamma_1$ into $\Gamma_2$ such that, every UP in $\Gamma_2$ has both terms pure wrt the same theory, $=_{E_1}$ or $=_{E_2}$.

\paragraph*{Step 3. (Variable identification)} Next, BSCA partitions $\Vars(\Gamma_2)$ into a partition $\VarIdP$ such that, each variable in $\Gamma_2$ is replaced with a representative from the same equivalence class in $\VarIdP$. The result is $\Gamma_3$.

\paragraph*{Step 4. (Split the problem)} The next step of BSCA is to split $\Gamma_3$ into two  UPs $\Gamma_{4.1}$ and $\Gamma_{4.2}$ such that, each set has every problem with terms that are pure wrt $=_{E_1}$ or $=_{E_2}$ respectively.

\paragraph*{Step 5. (Solve systems)} The penultimate step of BSCA is to partition all the variables in $\Gamma_3$ into a size of two: Let $p = \{ V_1, V_2 \}$ is a partition of $\Vars(\Gamma_3)$. Then, the earlier problems ($\Gamma_{4.1}$, $\Gamma_{4.2}$) are further split such that, all the variables in one set of the partition are replaced with new constants in the other set and  vice-versa. The resulting sets are $\Gamma_{5.1}$ and $\Gamma_{5.2}$.

\paragraph*{Step 6. (Combine unifiers)} The final step of BSCA is to combine the unifiers for $\Gamma_{5.1}$ and $\Gamma_{5.2}$, obtained using $A_{E_1}$ and $A_{E_2}$:

\begin{definition}{\em\bf [Combined Unifier]}\label{d.Combined-Unifier}

Let $\Gamma$ be a $E$-UP where $(E_1 \cup E_2) = E$. Let $\sigma_i \in A_{E_i}(\Gamma_{5.i})$, $i \in \{1,2\}$ and let $V_i = \Vars(\Gamma_{5.i})$, $i \in \{1, 2 \}$.

Suppose `$<$' is a linear order on $\Vars(\Gamma)$ such that $Y < X$ if $X$ is not a subterm of an instantiation of $Y$:

	\[	(\fa X, Y \in \Vars(\Gamma))((Y < X) \Ra (\not\ex \sigma)(X \sqsubset Y\sigma)).  \]

Let $\msf{least}(X,T,<)$ be defined as the minimal element of set $T$, when ordered linearly by the relation `$<$'. i.e., 

\[	\msf{least}(X,T,<)	\Lra	(\fa Y \in T)((Y \neq X) \Ra (X < Y)).	\]

Then, the combined UA for $\Gamma$, namely $A_{E_1 \cup E_2}$, is defined such that,

\[	A_{E_1\cup E_2}(\Gamma) = \{ \sigma \mid (\ex \sigma_1,\sigma_2)((\sigma = \sigma_1 																					\odot \sigma_2) \wedge (\sigma_1 \in A_{E_1}(\Gamma_{5.1})) \wedge (\sigma_2 \in A_{E_2}(\Gamma_{5.2}))) \}.
\] 

\noindent
where, if $\sigma = \sigma_1 \odot \sigma_2$, then,

\begin{itemize}

\item The substitution in $\sigma$ for the least variable in $V_1$ and $V_2$ is  from $\sigma_1$ and $\sigma_2$ respectively: \\

$(\fa i \in \{ 1, 2 \})( (X \in V_i) \wedge \msf{least}(X, \Vars(\Gamma), <)  \Rightarrow (X\sigma = X\sigma_i))$; and \\

\item For all other variables $X$, where each $Y$ with $Y < X$ has a substitution already defined, define
$X\sigma = X\sigma_i\sigma$ $(i \in \{1,2\})$: \\

$(\fa i \in \{ 1, 2 \})( (\fa X \in V_i)( (\fa Y)( (Y < X) \wedge (\ex Z)(Z/Y \in \sigma) )) \Rightarrow (X\sigma = X \sigma_i \sigma))$.

\end{itemize}

\end{definition}

It has been proven in~\cite{BS96} that the combination algorithm defined above is a $(E_1 \cup E_2)$-UA for any $(E_1 \cup E_2)$-UP if $E_1$-UA and $E_2$-UA are known to exist and if $=_{E_1}$, $=_{E_2}$ are disjoint. The combination of \tstd{} and \tacun{} UAs which is of interest to us in this paper has been explained to be finitary (i.e., return a finite number of unifiers) when combined using BSCA~\cite{Tuengerthal-TR-2006}.

We now prove a simple lemma which states that the combined unifier of two unifiers is equal to one of the unifiers, if the other unifier is empty.

\begin{lemma}{\em\bf [Combined unifier when one of the unifier is empty]}\label{l.emptyT2sub}

Let $\Gamma, \sigma, \sigma_1, \sigma_2, V_1,  V_2$, and $<$ be as defined above in Def.~\ref{d.Combined-Unifier}. Then,

	\[	
(\sigma = \sigma_1 \odot \sigma_2) \wedge (\sigma_2 = \{ \}) \wedge (V_2 = \{ \}) \Ra (\sigma = \sigma_1).	\]
	
\end{lemma}

\begin{proof}

	Let $\Vars(\sigma) = \{X \mid \_/X \in \sigma\}$.
	
	From Def.~\ref{d.Combined-Unifier}, if $\sigma = \sigma_1 \odot \sigma_2$, then,
	
	\[
		(\fa i \in \{1,2\})((X \in V_i) \wedge \msf{least}(X,\Vars(\Gamma),<) \Ra (X\sigma = X\sigma_i)).	\]
	
	But since $\sigma_2 = \{ \}$ and $V_2 = \{\}$, we have,
	
	\begin{equation}\label{e.combi-unifier-least}	
		(\fa X \in V_1 \cup V_2)(\msf{least}(X,\Vars(\Gamma),<) \Ra (X\sigma = X\sigma_1)).	
	\end{equation}
	
	Also from Def.~\ref{d.Combined-Unifier}, 
	
\[	(\fa X \in V_1 \cup V_2)((\fa Y)((Y < X) \wedge (\ex Z)(Z/Y \in \sigma) \Ra (X\sigma = X\sigma_i\sigma))).  \]

Again, since $\sigma_2 = \{\}$  and $V_2 = \{\}$, this implies,

\begin{equation}\label{e.combi-unifier-othervars}
	(\fa X \in V_1 \cup V_2)((\fa Y)((Y < X) \wedge (\ex Z)(Z/Y \in \sigma) \Ra (X\sigma = X\sigma_1))).
\end{equation}

Combining (\ref{e.combi-unifier-least}) and (\ref{e.combi-unifier-othervars}), we have,

\begin{equation}\label{e.sigma-sigma1}
		(\fa X \in V_1 \cup V_2)(X\sigma = X\sigma_1).
\end{equation}

Further, since $\sigma_2 = \{ \}$, and $V_2 = \{ \}$, we have $\Vars(\sigma) = \Vars(\sigma_1) = V_1$ and hence, combining this with (\ref{e.sigma-sigma1}), we have $\sigma = \sigma_1$.

\end{proof}

\subsection{\tacun{}-UPs in $\nutsat{}$ protocols have only constants as subterms}\label{ss.ACUN-Probs-Only-Constants}


Our next lemma is a bit lengthy. This lemma is the lynchpin of the paper and forms the crux of our two main theorems in Section~\ref{s.main-results}.

It concerns combined UPs involving the disjoint theories, $=_{\std}$ and $=_{\acun}$. We prove that, if we follow BSCA for finding unifiers for a  $(\SUA)$-UP, say $\Gamma$, that do not have free variables inside \txor{} terms, the terms in all the \tacun{}-UPs ($\Gamma_{5.2}$) from those will always have only constants as subterms. Consequently, we will end up in an empty set of substitutions returned by the \tacun{}-UA for $\Gamma_{5.2}$, even when their  terms are equal in the $=_{\acun}$ theory.

\begin{lemma}{  {\em\bf [$\acun$-UPs have only constants as subterms]}}\label{l.acun-nosub}

Let $\Gamma = \{m \stackrel{?}{=}_{\SUA} t\}$ be a $(\SUA)$-UP that is $(\SUA)$-Unifiable, and where no subterm of $m$ or $t$ is an \txor{} term with free variables:

\[
		(\fa x)
		\left(
			\begin{array}{c}
				( (x \sqss m) \vee (x \sqss t) )  \wedge (n \in \mathbb{N})\\
				(x = x_1 \oplus \ldots \oplus x_n) \wedge (n > 1)
			\end{array}
			\Ra
			(\fa i \in \{1,\ldots,n\})(x_i \notin \Vars)
		\right).
\]

Then,

\[
(\fa m' \stackrel{?}{=}_{\acun} t' \in \Gamma_{5.2};x)
	\left(
	\begin{array}{c}
		\begin{array}{c}
			(x \sqsubset m') \vee (x \sqsubset t')			
		\end{array}
	\Ra
		(x \in \Constants)
	\end{array}
	\right).
\]

\end{lemma}

\begin{proof}

	Let   $\sigma$   be a set of substitutions s.t.   $\sigma \in A_{\SUA}(\Gamma)$.
	
	Then, from Def.~\ref{d.Combined-Unifier} ({\em\bf Combined Unifier}),   $\sigma \in \sigma_1 \odot \sigma_2$,   where   $\sigma_1 \in A_{\std}(\Gamma_{5.1})$   and   $\sigma_2 \in A_{\acun}(\Gamma_{5.2})$.
	
Suppose there is a term   $t$   in   $\Gamma$   with an alien subterm   $t'$ wrt the theory   $=_{\acun}$   (e.g.   $[1,n_a]^{\to}_k \oplus b \oplus c$   with the alien subterm of   $[1,n_a]^{\to}_k$).

Then, from the definition of $\Gamma_2$,   it must have been replaced with a new variable in   $\Gamma_2$.   i.e.,

\begin{equation}\label{e.replace-W-NewVars}	
(\fa t, t')
\left(
\left(
\begin{array}{c}
	(t \in \Gamma) \wedge (t = \_ \oplus \ldots \oplus \_) \wedge \\
	(t' \sqss t) \wedge \msf{ast}(t',t,=_{\acun})
\end{array}
\right)
\Ra
(\ex X)
\left(
\begin{array}{c}
	(X \stackrel{?}{=}_{\std} t' \in \Gamma_2) \wedge \\
	(X \in \NewVars)
\end{array}
\right)
\right).
\end{equation}

where $\NewVars \subset \Vars \setminus \Vars(\Gamma)$.

Since \txor{} terms do not have free variables from hypothesis, it implies that every free variable in an \txor{} term in $\Gamma_2$ is a new variable:

\begin{equation}\label{e.acun-terms-no-free-vars}
(\fa t, t')
\left(
\left(
	\begin{array}{c}
		(t \in \Gamma_2) \wedge 
		\pure(t, =_{\acun}) \wedge \\
		(t' \sqss t) \wedge (t' \in \Vars)
	\end{array}
\right)
	\Ra
	(t' \in \NewVars)
\right).
\end{equation}

Since every alien subterm of every term in   $\Gamma$   has been replaced with a new variable (\ref{e.replace-W-NewVars}), combining it with (\ref{e.acun-terms-no-free-vars}), \txor{}  terms in   $\Gamma_2$   must now have only constants and/or new variables as subterms:

\begin{equation}\label{e.newvars-or-constants-Gamma2}
(\fa t, t')
\left(
\left(
	\begin{array}{c}
		\pure(t,=_{\acun}) \wedge \\
		(t \in \Gamma_2) \wedge (t' \sqss t)
	\end{array}
\right)	
\Ra
	(t' \in \NewVars \cup \Constants)
	\right).
\end{equation}

Let   $\VarIdP$   be a partition of   $\Vars(\Gamma_2)$   and   $\Gamma_3 = \Gamma_2 \rho$,   such that \\

  $\Gamma_2 \rho    =    \{  s \stackrel{?}{=} p   \mid   (s \stackrel{?}{=} p := s'\rho \stackrel{?}{=} t'\rho) \wedge s' \stackrel{?}{=} t' \in \Gamma   \}$   \\

where   $\rho$   is the set of substitutions where each set of variables  in   $\VarIdP$   has been replaced with one of the variables in the set:

\[
\rho    = 
  \left\{
\begin{array}{cl}
x/X   \mid  
\left(
\begin{array}{c}
	 (\fa Y_1/X_1, Y_2/X_2 \in \rho;	\vip \in \mi{VarIdP}) 
				 \left(
				 	 \begin{array}{c}
					 	 ( X_1, X_2 \in \vip) \Ra \\
					 	 (Y_1 = Y_2) \wedge \\
					 	 (Y_1, Y_2 \in \vip)  
				 	 \end{array}
				 \right)
\end{array}
\right)
\end{array}
\right\}.
\]

Can there exist a substitution   $X/Y$   in   $\rho$   such that   $Y \in \NewVars$ and   $X \in \Vars(\Gamma)$?

To find out, consider the following two statements:

\begin{itemize}
	
		\item From (\ref{e.replace-W-NewVars}), every new variable   $Y$   in   $\Gamma_2$   belongs to a $\std$-UP in  $\Gamma_2$:
		
		\[	(\fa Y \in \NewVars)( (Y \in \Vars(\Gamma_2) \Ra (\ex t)(\pure(t, =_{\std}) \wedge Y \stackrel{?}{=}_{\std} t \in \Gamma_2)) ).  \]

		\item Further, from hypothesis, we have that \txor{} terms in   $\Gamma$   do not have free variables. Hence, every free variable is a proper subterm\footnote{$t$ is a proper subterm of $t'$ if $t \sqsubset t' \wedge t \neq t'$.} of a purely $=_{\std}$ term:

\[
(\fa X \in \Vars(\Gamma))
\left(
			\begin{array}{c}
					(\ex t \in \Gamma) ( (X \sqss t) \wedge 
					\pure(t, =_{\std}) \wedge (X \neq t))
			\end{array}
\right).
\]

\end{itemize}

The above two statements are contradictory: It is not possible that a new variable and an existing variable can be replaced with each other, since one belongs to a $\std$-UP, and another is always a proper subterm of a term that belongs to a $\std$-UP. 

Hence,   $\VarIdP$   cannot consist of sets where new variables are replaced by   $\Vars(\Gamma)$.   i.e.,

\begin{equation}\label{e.VarIdP-Prop}
	(\nexists X, Y; \vip \in \VarIdP)
	\left(
		\begin{array}{c}
			(Y,X \in \vip) \wedge (Y \in \NewVars) \wedge \\
			(X \in \Vars(\Gamma)) \wedge (X/Y \in \rho) 
		\end{array}
	\right)
\end{equation}


Writing (\ref{e.VarIdP-Prop}) in (\ref{e.newvars-or-constants-Gamma2}), we have,

\begin{equation}\label{e.newvars-or-constants-Gamma3}
(\fa t, t')
\left(
\left(
	\begin{array}{c}
		\pure(t,=_{\acun}) \wedge \\
		(t \in \Gamma_3) \wedge (t' \sqss t)
	\end{array}
\right)	
\Ra
	(t' \in \NewVars \cup \Constants)
	\right).
\end{equation}

Further, if a variable belongs to a UP of $\Gamma_3$, then the other term of the UP is pure wrt $=_{\std}$ theory:
 
\begin{equation}\label{e.NewVars-in-Gamma3-mapto-STDterms}
(\fa X \in \Vars(\Gamma_3), t)
\left(
	\left(
		\begin{array}{c}
			(X \stackrel{?}{=}_{\std} t \in \Gamma_3)  \vee \\
			(t \stackrel{?}{=}_{\std} X \in \Gamma_3)
		\end{array}
	\right)
	\Ra
		(X \in \NewVars) \wedge
		\pure(t,=_{\std})
\right).
\end{equation}

Now suppose   $\Gamma_{4.2}    =    \{  s \stackrel{?}{=} t   \mid   ( s \stackrel{?}{=} t \in \Gamma_3) \wedge \pure(s,=_{\acun}) \wedge \pure(t,=_{\acun})  \}$, 
$\{ V_1, V_2 \}$   a partition of   $\Vars(\Gamma) \cup \NewVars$,   and

\[	\Gamma_{5.2}   =   \Gamma_{4.2} \beta, 	\]

where, $\beta$ is a set of substitutions of new constants to $V_1$: 

 \[	\beta = \{  x/X   \mid   (X \in V_1) \wedge (x \in \Constants \setminus (\Constants(\Gamma) \cup \Constants(\Gamma_{5.1})))  \}.	\]

From hypothesis,   $\Gamma_{5.2}$   is  $\acun$-Unifiable. Hence, we have:

\[	(\fa \sigma)((\fa m' \stackrel{?}{=}_{\acun} t' \in \Gamma_{5.2})(m'\sigma =_{\acun} t'\sigma) \Lra \sigma \in A_{\acun}(\Gamma_{5.2})).  \]

Now consider a   $\sigma$   s.t.   $\sigma \in A_{\acun}(\Gamma_{5.2})$.

From (\ref{e.newvars-or-constants-Gamma3}), we have that \txor{} terms in   $\Gamma_{5.2}$   have only new variables and/or constants and from (\ref{e.NewVars-in-Gamma3-mapto-STDterms}) we have that if $X \in \Vars(\Gamma_{5.2})$, then there exists $t$ s.t. $X \stackrel{?}{=}_{\std} t \in \Gamma_{5.1}$ and $t$ is pure wrt $=_{\std}$ theory.
 
Suppose   $V_2 \neq \{ \}$.   Then, there is at least one variable, say   $X \in \Vars(\Gamma_{5.2})$.   This implies that   $X$   is replaced with a constant (say  $x$)   in   $\Gamma_{5.1}$.

Since   $X$   is necessarily a new variable and one term of a $\std$-UP,   this implies that   $x$   must equal some compound term made with $\StdOps$.

However, a compound term made with $\StdOps$ can never equal a constant under the $=_{\std}$ theory:

\[	(\not\ex \mi{op} \in \StdOps; t_1, \ldots, t_n; x \in \Constants)(x =_{\std} \mi{op}(t_1,\ldots,t_n)),		\]

a contradiction.

Hence,   $\sigma = \{ \}$,   $V_2 = \{ \}$   and our hypothesis is true that all \txor{} terms in   $\Gamma_{5.2}$   necessarily contain only constants:

\[	(\fa m' \stackrel{?}{=}_{\acun} t' \in \Gamma_{5.2}; x)
			\left(
						\begin{array}{c}
								 (x \sqss m) \vee (x \sqss t) 
								\Ra (x \in \Constants)
						\end{array}
			\right).
\]

\end{proof}

\section{Main Results}\label{s.main-results}

In this section, we will prove our main results. We will first prove that $\nutsat$~protocols are not susceptible to type-flaw attacks in Section~\ref{ss.type-flaw-result}. We will then prove that $\munutsat$~protocols are not susceptible to multi-protocol attacks in Section~\ref{ss.multi-prot-result}.

\subsection{{\sf NUT} prevents type-flaw attacks}\label{ss.type-flaw-result}

We will now prove our first main result that $\nutsat{}$ protocols will not have any type-flaw attacks. The main idea is to show that every unification when solving a constraint sequence from a $\nutsat{}$ protocol results in a well-typed unifier. We follow the outline below:

\begin{enumerate}

	\item We will first establish that normal constraint sequences from $\nutsat{}$ protocols do not contain variables in the target or term set of their active constraints (either freely or inside \txor{} terms), but only subterms of the initial term set;
	
	\item We then infer from Lemma~\ref{l.acun-nosub} that if a $(\SUA)$-UP, say $\Gamma$, does not have free variables inside $\txor{}$ terms, then terms in it's $\Gamma_{5.2}$ will have only constants as subterms;
	
	\item Next, we infer in Lemma~\ref{l.stdwellty} that UPs in $\Gamma_{5.1}$  unify only under well-typed substitutions, if they were created from the same underlying term of the protocol, by applying two well-typed substitutions (which is true for semi-bundles from $\nutsat{}$ protocols, under Assumption~\ref{a.wellty});
	
	\item Finally, the combined unifier for $\Gamma$ is simply the unifier for $\Gamma_{5.1}$, from Lemma~\ref{l.emptyT2sub} ({\bf Combined unifier when one of the unifier is empty}), and hence is always well-typed.

\end{enumerate}
\begin{theorem}\label{t.type-flaw}
	$\nutsat{}$ protocols are secure against type-flaw attacks in the $=_{\SUA}$ theory.
\end{theorem}

\begin{proof}

	From Def.~\ref{d.type-flaw-attack} ({\bf type-flaw attacks}),  a protocol is susceptible to type-flaw attacks if a constraint sequence from a semi-bundle of the protocol is satisfiable only with a substitution that is not well-typed.
	
	We will show that this never happens; i.e., every satisfiable constraint sequence from a semi-bundle of a $\nutsat{}$ protocol is satisfiable only with a well-typed substitution.

	Let $P$ be a $\nutsat{}$ protocol, $S$ a semi-bundle from $P$, and $cs$ a constraint sequence from $S$. 	Suppose $cs$ is satisfiable with a substitution  in the $=_{\SUA}$ theory. i.e.,
	
\begin{equation}\label{e.init-stmt}
	\semibundle(S,P) \wedge \conseq(cs,S) \wedge \satisfiable(cs,\_,=_{\SUA}). 
\end{equation}

	From (\ref{e.satisfiable}) ({\bf satisfiability}), suppose we have $r_1,\ldots,r_n \in \Rules$ s.t. 

\begin{equation}\label{e.infer-from-satisfiability}
\begin{array}{c}
\left(
	 \begin{array}{l}
	 	 \appl(r_1,cs,cs_1,\{\}, \sigma_1,=_{\SUA}) \wedge \\
  		\appl(r_2,cs'_1,cs_2,\sigma_1,\sigma_2,=_{\SUA}) \wedge \ldots \\
    	\appl(r_n,cs'_{n-1},cs_n,\sigma_{n-1},\sigma_n,=_{\SUA}) \wedge \\
    	\msf{simple}(cs_n) \wedge\\
   		(\forall i \in \{ 1,\ldots,n\})(cs'_i = \normalize(cs_i))
   \end{array}
 \right).
\end{array}
\end{equation}

Now every $cs'_i$ in (\ref{e.infer-from-satisfiability}) is normalized. Hence, their active constraints do not have variables in the targets or term sets. Further, since $P$ is $\nutsat{}$, no term of the form $t_1\oplus \ldots \oplus t_p$ ($p > 1$) can have a free variable in the set $\{t_1, \ldots, t_p\}$ (from \tnut{} Condition~3). i.e.,

\begin{equation}\label{e.infer-from-lemma5}
\begin{array}{c}
(\fa i \in \{1,\ldots,n\};x)
	\begin{array}{c}	
 				  \left(
						\begin{array}{l}
							\act(m:T,cs'_i) \wedge (p > 1) \\
							(x = m) \vee (x \in \{t_1, \ldots, t_p\}) \wedge \\ 
							(t_1 \oplus \ldots \oplus t_p \in m \cup T)			
						\end{array}
					\right)
					\Ra  
					\left(
						\begin{array}{l}
							(x \notin \Vars(S)) \wedge \\
							(x \in \SubTerms(S\sigma_i))
						\end{array}
					\right).
		\end{array}	
\end{array}	
\end{equation}

From the set $\mi{Rules}$ it is clear that only $\msf{un}$ and $\msf{ksub}$ potentially change the set of substitutions, when applied to a constraint sequence. i.e.,

\begin{equation}\label{e.only-unksub}
	(\forall r \in \Rules)(\appl(r,\_,\_,\sigma,\sigma',\_) \wedge (\sigma \subset \sigma') \Rightarrow (r = \msf{un}) \vee (r = \msf{ksub})).   
\end{equation}

\noindent
Consider rules $\msf{un}$ and $\msf{ksub}$:

\[
	\appl(\msf{un},cs,cs',\sigma,\sigma',=_{\SUA})  
	\Lra
	(\ex m, T, t)
	\left(  
	\begin{array}{l}
	\act(m:T \cup t, cs) \wedge 
	 (cs' = cs_< \tau^{\frown}cs_>\tau) \wedge \\
	 (\sigma' = \sigma \cup \tau) \wedge (\tau \in U_{\SUA}(\{m  \stackrel{?}{=}_{\SUA} t\}))
	\end{array}
	 \right)
\]

\[
	\appl(\msf{ksub},cs,cs',\sigma,\sigma',=_{\SUA})  
	\Lra
	(\ex m, T, t)
	\left(  
	\begin{array}{l}
	\act(m : T \cup [t]^{\to}_k, cs) \wedge \\
	 (cs' = cs_< \tau^{\frown}[m\tau : T\tau \cup [t]^{\to}_k\tau]^{\frown} 
	  cs_>\tau) \wedge \\
	 (\sigma' = \sigma \cup \tau) \wedge (\tau \in U_{\SUA}(\{k \stackrel{?}{=}_{\SUA} \mi{pk}(\epsilon)\}))
	\end{array}
	 \right)
\]

Suppose $\Gamma = \{ m \stackrel{?}{=}_{\SUA} t \}$ where $m = m'\sigma_m\sigma$ and $t = t'\sigma_t\sigma$ and for some $r, r' \in P$, $r\sigma_m, r\sigma_t \in S$.

Suppose $\tau \in U_{\SUA}(\Gamma)$. Then, using Def.~\ref{d.Combined-Unifier} ({\bf Combined Unifier}),  let $\tau \in \tau_{\std} \odot \tau_{\acun}$  where  $\tau_{\std} \in A_{\std}(\Gamma_{5.1})$  and  $\tau_{\acun} \in A_{\acun}(\Gamma_{5.2})$.

From (\ref{e.infer-from-lemma5}), we can infer that the conditions of Lemma~\ref{l.acun-nosub} ({\bf $\acun$ UPs have only constants}) are met:

\begin{equation}\label{e.xorterms-have-novars}
		(\fa x)
		\left(
			\begin{array}{c}
				((x \sqss m) \vee (x \sqss t)) \wedge\\
				(x = x_1 \oplus \ldots \oplus x_n) \wedge (n > 1)
			\end{array}
		\right)
		\Ra
			(\fa i \in \{1,\ldots,n\})(x_i \notin \Vars).
\end{equation}

And therefore, we infer from Lemma~\ref{l.acun-nosub} that:

	\begin{equation}\label{e.Th1-ACUN-Empty-Unif}
		(\fa \Gamma_{5.2};\tau_{\acun} \in A_{\acun}(\Gamma_{5.2}))(\tau_{\acun} = \{\}).  		\end{equation}

Now consider problems in $\Gamma_{5.1}$. Suppose $\tuple{m_1}{t_1} \in \Gamma_{5.1}$. Let $m_1 = x\sigma_m\sigma\rho\alpha$ and $t_1 = y\sigma_t\sigma\rho\alpha$, where $\sigma$ is as defined in rule $\msf{un}$; $x, y \in \SubTerms(P)$; $\rho$ as defined in Lemma~\ref{l.acun-nosub} and $\alpha$ is a set of substitutions s.t.

\[	\Gamma_{5.1}   =   \Gamma_{4.1} \alpha, 	\]

where, $\alpha$ substitutes new constants to $V_2$: 

 \[	\alpha = \{  x/X   \mid   (X \in \Vars(\Gamma_{5.2})) \wedge (x \in \Constants \setminus \Constants(\Gamma))  \}.	\]

From Lemma~\ref{l.acun-nosub}, we have that $\Vars(\Gamma_{5.2}) = \{ \}$. Hence, $\alpha = \{ \}$. 

Also from Lemma~\ref{l.acun-nosub}, we have that, whenever $\Gamma$ is $(\SUA)$-Unifiable, $\Gamma_{4.2}$ will not have any variables of $\Gamma$, and $\Gamma_{5.2}$ will not have any variables at all. Hence, we have that every partition of $\VarIdP$ (defined in Lemma~\ref{l.acun-nosub}) in which there is a variable of $\Gamma$, has only that variable and no others in the partition:

\begin{equation}\label{e.single-ex-vars-in-varidp}
	(\fa \vip \in \VarIdP; X, Y \in \vip)(X \in \Vars(\Gamma) \Ra X = Y).
\end{equation}

Now, $\Vars(\Gamma_{5.1}) = \Vars(\Gamma) \cup \NewVars$. 

From (\ref{e.single-ex-vars-in-varidp}), we have,

\begin{equation}\label{e.ex-vars-wellty}
	(\fa x/X \in \rho)(X \in \Vars(\Gamma) \Ra \wellty(x/X)).
\end{equation}

Now,

\begin{itemize}

	\item From (\ref{e.infer-from-lemma5}), we have that $m, t \in \SubTerms(S\sigma)$;

	\item From BSCA, if $m, t \in \SubTerms(S\sigma)$, and $m_1, t_1 \notin \NewVars$, then $m_1$ and $t_1$ must belong to $\SubTerms(S\sigma\rho)$;

	\item From \tnut{} Conditions~1 and 3, if $m_1$ is $\std$-Unifiable with $t_1$, then $x$ must equal $y$.
	
\end{itemize}

	If $x = y$, since $\wellty(\sigma_m)$ and $\wellty(\sigma_t)$ from Assumption~\ref{a.wellty} ({\bf Honest agent substitutions are always well-typed}), assuming $\wellty(\sigma)$, and $\wellty(\rho)$ from~(\ref{e.ex-vars-wellty}), we can infer from Lemma~\ref{l.stdwellty} ({\bf Well-typed $\std$ terms unify only under well-typed unifiers}) that, $\wellty(\delta)$, where $m_1\delta =_{\std} t_2\delta$, if none of $\NewVars$ exist as subterms of $m_1$ or $t_1$:
	
\begin{equation}\label{e.only-ex-vars-wellty}
	(\fa m_1 \stackrel{?}{=}_{\std} t_1 \in \Gamma_{5.1})
	\left(
		\begin{array}{c}
			(\NewVars \cap \SubTerms(\{m_1,t_1\}) = \{ \}) \\
			\wedge  (m_1\delta 			=_{\std} t_1\delta)
		\end{array}
	\Ra 
		\wellty(\delta)
	\right).
\end{equation}
	
But what if $m_1$ or $t_1$ contain new variables as subterms?

Now the type of the new variables is the type of compound terms that they replace:

\[
		(\fa X \in \NewVars)(X \stackrel{?}{=}_{\acun} t \in \Gamma_2 \Ra \type(X) = \type(t)).
\]

Suppose $X/Y \in \rho$, where $X, Y \in \NewVars$ (note that $X$ or $Y$ cannot belong to $\Vars(\Gamma)$ from equation (\ref{e.VarIdP-Prop}) in Lemma~\ref{l.acun-nosub}).

Suppose there exist some $t_1, t_2$ such that $t_1 \stackrel{?}{=}_{\std} X$ belongs to $\Gamma_{5.1}$ and $t_2 \stackrel{?}{=}_{\std} Y$ belongs to $\Gamma_{5.1}$. Suppose $t_1, t_2$ do not have any new variables as subterms. Then, from (\ref{e.only-ex-vars-wellty}), we have $\wellty(\theta)$, where $t_1\theta =_{\std} t_2\theta$, and hence, we have $\wellty(X/Y)$:

\begin{equation}\label{e.new-vars-also-wellty}
(\fa X, Y \in \NewVars)((X/Y \in \rho) \Ra \wellty(X/Y)).
\end{equation}

Combining (\ref{e.ex-vars-wellty}) and (\ref{e.new-vars-also-wellty}), we have, $\wellty(\rho)$.

Given this, using induction on terms, we conclude similar to concluding (\ref{e.only-ex-vars-wellty}) that every problem in $\Gamma_{5.1}$ unifies under a well-typed substitution: 

\[
	(\fa m_1 \stackrel{?}{=}_{\std} t_1 \in \Gamma_{5.1})((m_1\tau_{\std} =_{\std} t_1\tau_{\std}) \Ra \wellty(\tau_{\std})).	
\]

Now,

\begin{tabbing}
		$\tau$  \=  $=$ \= $\tau_{\std} \odot \tau_{\acun}$  \\
		\> = \> $\tau_{\std} \odot \{\}$ (from~\ref{e.Th1-ACUN-Empty-Unif})  \\
		\> = \> $\tau_{\std}$. (from Lemma~\ref{l.emptyT2sub} ({\bf Combined unifier when one of the unifier is empty}))
\end{tabbing}

Since $\wellty{}(\tau_{\std})$ from above, this implies, $\wellty{}(\tau)$.

Similarly, for $\msf{ksub}$, we can conclude, $\wellty(\tau)$, where $\tau \in A_{\SUA}(\{k \stackrel{?}{=}_{\SUA} \pk(\epsilon)\})$, provided $k$ is not a variable, and indeed it is not by \tnut{} Condition~2.

So the only rules that potentially change the substitution ($\msf{un}$, $\msf{ksub}$) produce well-typed substitutions. We can apply this in (\ref{e.only-unksub}) and write:

\begin{equation}
	(\forall r \in \{r_1,\ldots,r_n\})
	\left(
		\begin{array}{c}
			\left(
				\begin{array}{c}
					\appl(r,\_,\_,\sigma,\sigma',=_{\SUA}) \wedge \\
					\wellty(\sigma)
				\end{array}
			\right)		
			\Rightarrow 
			\wellty(\sigma')
		\end{array}
	\right).   
\end{equation}

Since all other rules except $\msf{un}$ and $\msf{ksub}$ do not change the attacker substitution, we can combine the above statement with (\ref{e.infer-from-satisfiability}) and conclude:

\begin{equation}\label{e.sat-wellty}
\left(
\begin{array}{l}
 \appl(r_1,cs,cs_1,\{\}, \sigma_1, =_{\SUA}) \wedge \\
  	\appl(r_2,cs'_1,cs_2,\sigma_1,\sigma_2, =_{\SUA}) \wedge \ldots \\
    \appl(r_n,cs'_{n-1},cs_n,\sigma_{n-1},\sigma_n, =_{\SUA}) \wedge \\
    \msf{simple}(cs_n) \wedge\\
   	(\forall i \in \{ 1,\ldots,n\})(cs'_i = \normalize(cs_i))
   \end{array}
   \right)
   \Ra
   \wellty(\sigma_n).
\end{equation}

(Note that we concluded $\wellty(\tau_{\std})$ assuming that $\sigma$ in rule $\msf{un}$ was well-typed. Thus, in (\ref{e.sat-wellty}), $\sigma_1$ is well-typed and inductively, all of $\sigma_2, \ldots, \sigma_n$ are well-typed).

Finally, we can combine the above statement with (\ref{e.init-stmt}) and form:

\[
	(\fa cs, S, \sigma)
	\left(
	\left(
	\begin{array}{c}
  	\semibundle(S,P)   \wedge \conseq(\msf{cs},S) \wedge \\
   	\satisfiable(\msf{cs},\sigma, =_{\SUA}) 
	\end{array}
	\right)
	\Ra
	\wellty(\sigma)
	\right).
\]

From Def.~\ref{d.type-flaw-attack} ({\bf type-flaw attack}), this implies, 

\[	\neg\msf{typeFlawAttack}(P, =_{\SUA}).	\]

Since we started out assuming that $P$ is a $\nutsat{}$ protocol, we sum up noting that $\nutsat{}$ protocols are not susceptible to type-flaw attacks.

\end{proof}

\subsection{$\mu$-{\sf NUT} prevents multi-protocol attacks}\label{ss.multi-prot-result}

We will now prove that $\munutsat$ protocols are not susceptible to multi-protocol attacks. 

The idea is to show that if a protocol is secure in isolation, then it is in combination with other protocols with which it is $\munutsat$.

To show this, we will achieve a contradiction by attempting to prove the contrapositive. i.e., if there is a breach of secrecy for a protocol in combination with another protocol with which it is $\munutsat$, then it must also have a breach of secrecy in isolation.

We will follow the outline below:

\begin{enumerate}

	\item We will first form a constraint sequence from a 						semi-bundle that has semi-strands from the 									combination of a secure 						protocol and 						another protocol with which it is 													$\munutsat$;
	\item We will then form another sequence that can be formed solely from a 						semi-bundle of the secure protocol by 											extracting it from the constraint sequence of the 					combination of semi-bundles;
	\item Finally, we will show that any reduction rules to 					satisfy the former resulting in a breach of 								secrecy 				can be equally applied on the 							latter, resulting in 					a breach of secrecy 					in it as well (thereby 									achieving a 				contradiction).

\end{enumerate}

We are now ready to prove our second main theorem.

\begin{theorem}\label{t.multi-prot}

	In the $=_{\SUA}$ theory, if a protocol is secure for secrecy, then it remains so in combination with any other protocol with which it is  $\munutsat$.

\end{theorem}

\begin{proof}

Suppose  $P_1$  is a protocol that is secure for secrecy in isolation in the $=_{\SUA}$ theory. i.e.,

\begin{equation}\label{e.sfs-S1}
	 \sfs(P_1, =_{\SUA}). 
\end{equation}
	 
Consider another protocol  $P_2$  such that,  $\munutsat(P_1,P_2)$.  Let,  $S_1$  and  $S_2$  be two semi-bundles from  $P_1$  and  $P_2$  respectively: 

\begin{equation}\label{e.S1P1-S2P2}
	\semibundle(S_1,P_1)	\wedge  \semibundle(S_2,P_2).
\end{equation} 

Consider a constraint sequence  $\mi{combcs}$  from  $S_{\mi{comb}} = S_1 \cup S_2$.  i.e., $\conseq(\combcs, S_{\mi{comb}})$. Consider another constraint sequence  $\isocs$, where, 

{\bf (a)} Targets in  $\combcs$  are targets in  $\isocs$  if the targets belong to  $S_1$:

\begin{equation}\label{e.same-targets}
		(\fa m : \_~\msf{in}~\combcs)( (m \in \mi{Terms}(S_1)) \Ra (m : \_~\msf{in}~ \isocs) ).
\end{equation}
 
{\bf (b)} Term sets in  $\combcs$  are term sets in  $\isocs$  but without terms from  $S_2$:
	
\begin{equation}\label{e.same-termsets}
		\left(
			\begin{array}{l}
				\fa m_1 : T_1, \\
				    m_2 : T_2 ~\msf{in} ~\combcs
			\end{array}
		\right)
		\left(
			\begin{array}{c}
				m_1 : T_1 \prec_{\combcs} m_2 : T_2 \\
				\Ra \\
				(\ex T'_1, T'_2)
				\left(
					\begin{array}{c}
							(m_1 : T'_1 \prec_{\isocs} m_1 : T'_2) \wedge \\
							(T'_1 = T_1 \setminus T''_1) \wedge (T'_2 = T_2 \setminus 																														T''_2) \\
							(\fa t \in T''_1 \cup T''_2)(t \in \SubTerms(S_2))
					\end{array}
				\right)	
			\end{array}
		\right).
\end{equation}

Then, from Def.~\ref{d.constraints} ({\bf Constraints}) we have that $\isocs$ is a constraint sequence from $S_1$ alone. i.e., $\conseq(\isocs,S_1)$.

Suppose  $\combcs$  and  $\isocs$  are normalized. To achieve a contradiction, let there be a violation of secrecy in  $S_{\mi{comb}}$  s.t.  $\combcs$  is satisfiable after an artificial constraint with a secret constant of  $S_1$, say $\mi{sec}$,  is added to it:

\begin{equation}
		(\combcs = [ \_  :  \_, \ldots, \_  :  T ]) \wedge 								 			\satisfiable(\combcs^{\frown}[\mi{sec}  :  T], \_, =_{\SUA}). 
\end{equation}

Suppose $[r_1, \ldots, r_n] = R$, such that $r_1,\ldots,r_n \in \Rules$. Then, from the definition of satisfiability (\ref{e.satisfiable}), using  $R$,   say we have:

\begin{equation}\label{e.comb-satisfiable}
\left(
	\begin{array}{l}
	 (\combcs = [\_ : \_, \ldots, \_ : T]) \wedge \\
	 \appl(r_1,\combcs^{\frown}[\mi{sec} : T],\combcs_1,\{\}, \sigma_1, =_{\SUA})  \wedge \\
  	\appl(r_2,\combcs'_1,\combcs_2,\sigma_1,\sigma_2, =_{\SUA})  \wedge  \ldots \wedge \\
    \appl(r_n,\combcs'_{n-1},\combcs_n,\sigma_{n-1},\sigma_n, =_{\SUA})  \wedge \\
    \simple(\combcs_n)  \wedge 
   (\fa i \in \{  1,\ldots,n \})(\combcs'_i = \normalize(\combcs_i))  
\end{array}
\right).
\end{equation}

From their descriptions, every rule in $\Rules$ adds subterms of existing terms (if any) in the target or term set of the active constraint:

\begin{equation}\label{e.rules-add-subterms}
\left(
	\begin{array}{c}
		 \appl(\_,cs,cs',\_,\_,\_) \wedge	\act(m : T, cs) \wedge \\
	   \act(m' : T', cs')  \wedge (x \in T' \cup m') 
\end{array}
\right)
  	    \Ra 
  	  (x \in \SubTerms(T \cup m)).
\end{equation}

Since every $\combcs'_i$ ($i = 1$ to $n$) in (\ref{e.comb-satisfiable}) is normalized, and since $P_1$ and $P_2$ are $\munutsat$, we have that no \txor{} term in the target or term sets of any of $\combcs'_i$ ($i = 1~\text{to}~n$) have free variables:

\begin{equation}\label{e.xor-terms-no-vars}
	(\fa i \in \{1,\ldots,n\})
		\left(
			\begin{array}{c}
				\act(m : T, \combcs'_i) \wedge (p \in \mathbb{N}) \wedge \\
				(t_1 \oplus \ldots \oplus t_p \in T \cup m) \Ra \\
				(\fa j \in \{1,\ldots,p\})(t_j \notin \Vars)
			\end{array}
	\right).
\end{equation}

Suppose  $\chcombcs$  is a normal, child constraint sequence of  $\combcs$  and  $\chisocs$  is a normal, child constraint sequence of  $\isocs$.

$\msf{un}$  and  $\msf{ksub}$ are the only rules that affect the attacker substitution.  We will show that these are equally applicable on $\chcombcs$ and $\chisocs$. Suppose:

\begin{itemize}

	\item $\Gamma = \{ m \stackrel{?}{=}_{\SUA} t \}$, is a $(\SUA)$-{\sf UP} 					and suppose  $m = m' \subcomb$, $t = t'\subcomb$, where  $m' \in \SubTerms(S_1)$;

	\item Variables in $\subcomb$ are substituted with terms from the same 				semi-bundle: 
	
				\begin{equation}\label{e.disj-subs}
					(\fa x/X \in \subcomb)((\ex i \in \{1,2\})(x,X \in \SubTerms(S_i))).
				\end{equation}

(This is vacuously true if $\msf{un}$ or $\msf{ksub}$ were never applied on $\combcs$, to derive $\chcombcs$, since $\subcomb$ is then empty).

	\item $\Gamma$  is $(\SUA)$-Unifiable. 

\end{itemize}

Let $\tau \in A_{\SUA} (\Gamma)$. Then, from Def.~\ref{d.Combined-Unifier} ({\bf Combined Unifier}), $\tau \in \tau_{\std} \odot \tau_{\acun}$, where $\tau_{\std} \in A_{\std}(\Gamma_{5.1})$ and $\tau_{\acun} \in A_{\acun}(\Gamma_{5.2})$.

Now from BSCA, if  $m_1 \stackrel{?}{=}_{\std} t_1 \in \Gamma_{5.1}$, and $\theta \in U_{\std}(\{m_1 \stackrel{?}{=}_{\std} t_1\})$, then we have the following cases:

\paragraph*{Variables.} If  $m_1$, and/or  $t_1$  are variables, from (\ref{e.xor-terms-no-vars}) and BSCA, they are necessarily new i.e., $m_1, t_1 \in \Vars \setminus \Vars(\Gamma)$ (unless  $m$  and  $t$  are variables, which they are not, since  $\chcombcs$  is normal). Hence, there are no new substitutions in $\theta$ to $\Vars(\Gamma)$ in this case.

\paragraph*{Constants.} If $m_1 \in \Constants(S_1)$, again from BSCA, $t_1$ cannot belong to $\Vars$, and it must be a constant. If $m_1$ is a fresh constant of $S_1$, then $t_1$ must also belong to $S_1$ from Assumption~\ref{a.freshness} ({\bf freshness}) and (\ref{e.disj-subs}), and if $m_1$ is not fresh, $t_1$ could belong to either $\SubTerms(S_1)$ or $\iik$ from Assumption~\ref{a.IIK2}. Further, $\theta = \{ \}$.
	
\paragraph*{Public Keys.} If $m_1 = \pk(\_)$, then $t_1$ must be some $\pk(\_)$  as well. From BSCA, $m_1$ cannot be such that  $\penc{\_}{m_1} \sqsubset m$. Further, there cannot be an \txor{} term, say $\ldots \oplus m_1 \oplus \ldots$ that is a subterm of $m$, from $\munut{}$ Condition~2. The only other possibility is that $m = m_1$. In that case, $t$ must also equal $t_1$, whence, $t$ can belong to $\iik$ from assumption~\ref{a.IIK1} ({\bf Intruder possesses all public-keys}). Hence, we have that, $(\fa x/X \in \theta)((\ex i \in \{1,2\})(x, X \in \SubTerms(S_i)))$. 

\paragraph*{Shared keys.} $m_1$ cannot be a long-term shared-key; i.e., $m_1 \neq \sh(\_,\_)$, since from Assumptions~\ref{a.LTKeys} and~\ref{a.Key-Var}, they do not appear as interms and from the definition of $\Gamma_{5.1}$, $m_1$ is necessarily an interm.
	
\paragraph*{Encrypted Subterms.} 	Suppose $m_1 = m_{11}\subcomb\rho$, $t_1 = t_{11}\subcomb\rho$, where $m_{11}, t_{11} \in \EnCp(S_1 \cup S_2)$ and $\rho$ is a set of substitutions from $\VarIdP$ defined in Lemma~\ref{l.acun-nosub}.
Then, from $\mu$-{\sf NUT} Condition~1 and (\ref{e.rules-add-subterms}), we have, $m_{11}, t_{11} \in 	\EnCp(S_i)$, where $i \in \{1,2\}$. Hence, $(\fa x/X \in \theta)((\ex i \in \{1,2\})(x,X \in \SubTerms(S_i)))$. 			

\paragraph*{Sequences.} If $m_1$ is a sequence, either $m$ must be a sequence, or there must be some $\ldots \oplus m_1 \oplus \ldots$ belonging to $\SubTerms(\{m,t\})$, from BSCA. But $m$ and $t$ cannot be sequences, since $\chcombcs$ is normal. Hence, by $\munut{}$ Condition~2 and (\ref{e.rules-add-subterms}), $m_1, t_1 \in \SubTerms(S_i)\subcomb\rho$, $i \in \{1,2\}$ and $(\fa x/X \in \theta)((\ex i \in \{1,2\})(x,X \in \SubTerms(S_i)))$. 			

In summary, we make the following observations about problems in $\Gamma_{5.1}$.

If $m_1$ is an instantiation of a subterm in $S_1$, then so is $t_1$, or $t_1$ belongs to $\iik$:

\begin{equation}\label{e.t1-belongs-to-S1}
		(\fa m_1 \stackrel{?}{=} t_1 \in \Gamma_{5.1})(m_1 \in \SubTerms(S_1)\subcomb\rho \Ra t_1 \in \SubTerms(S_1)\subcomb\rho \cup \iik).
\end{equation}

Every substitution in $\tau_{\std}$ has both its term and variable from the same semi-bundle:

\begin{equation}\label{e.tau-std-subs}
	(\fa x/X \in \tau_{\std})((\ex i \in \{1,2\})(x,X \in \SubTerms(S_i))).	
\end{equation}

Now consider the UPs in $\Gamma_{5.2}$. Applying (\ref{e.xor-terms-no-vars}) into  Lemma~\ref{l.acun-nosub}, we have that $\tau_{\acun} = \{ \}$. Combining this with (\ref{e.tau-std-subs}), we have: 

\begin{equation}\label{e.tau-pure} (\fa x/X \in \tau)((\ex i \in \{1,2\})(x,X \in \SubTerms(S_i)\subcomb)). 
\end{equation}

Suppose $m = m_1 \oplus \ldots \oplus m_p$ and $t = t_1 \oplus \ldots \oplus t_q$; $p, q \ge 1$, $x = m\tau$, $y = t\tau$ and $m'' =_{\SUA} x$ where $m'' = m'_1 \oplus \ldots \oplus m'_{p'}$, s.t. $(\fa i, j \in \{1,\ldots,p'\})(i \neq j \Ra m'_i\tau \neq_{\SUA} m'_j \tau)$ and $t'' =_{\SUA} y$, where $t'' = t'_1 \oplus \ldots \oplus t'_{q'}$, s.t. $(\fa i,j \in \{1,\ldots,q'\})(i \neq j \Ra t'_i\tau \neq_{\SUA} t'_j\tau)$. Informally, this means that, no two terms in $\{m'_1,\ldots,m'_{p'}\}$ or $\{t'_1,\ldots,t'_{q'}\}$ can be canceled.

Now, $m\tau =_{\SUA} t\tau$ implies, $(\fa i \in \{1,\ldots,p'\})((\ex j \in \{1,\ldots,q'\})(m'_i\tau\rho =_{\std} t'_j\tau\rho))$ with $p' = q'$. 
From $(\ref{e.t1-belongs-to-S1})$ and $\munut$, this means that $m \in \SubTerms(S_1)\subcomb$ implies, $t$ also belongs to $\SubTerms(S_1)\subcomb$ or $\iik$.

Now since $\Vars(m') \cup \Vars(t') \subset \Vars(S_1)$, we have, $m'\subcomb = m'\subiso$, and $t'\subcomb = t'\subiso$, where $\subcomb = \subiso \cup \{x/X \mid x, X \in \SubTerms(S_2) \}$. Combining this with (\ref{e.tau-pure}), we have that, $m'\subcomb\tau =_{\SUA} t'\subcomb\tau \Ra m'\subiso\tau =_{\SUA} t'\subiso\tau$.

Combining these with (\ref{e.same-targets}) and (\ref{e.same-termsets}), we can now write: 
		\begin{equation}\label{e.un-applies-equally}
			(\fa \chcombcs, \chisocs)
			\left(
			\begin{array}{l}
				\appl(\msf{un},\chcombcs,\chcombcs',\subcomb,\subcomb',=_{\SUA}) \\
				\childseq(\chcombcs,\combcs,=_{\SUA}) \wedge \\
				\childseq(\chisocs,\isocs,=_{\SUA})  \Ra \\
				\appl(\msf{un},\chisocs,\chisocs',\subiso,\subiso',=_{\SUA}) 
			\end{array}
			\right).
		\end{equation}

where, the active constraint in $\chcombcs$ and $\chisocs$ only differ in the   term sets:

\[
		\left(  
	\begin{array}{l}
		\act(m : \_ \cup t, \combcs) \wedge 	\act(m : \_ \cup t, \isocs) \wedge \\
		 (\combcs' = \combcs_<\tau^{\frown}
	 	\combcs_>\tau)  \wedge  	 (\isocs' = \isocs_<\tau^{\frown} 
	 	\isocs_>\tau) \wedge \\
	 	(\subcomb' = \subcomb \cup \tau) \wedge 	 (\subiso' = \subiso \cup \tau) 			\wedge 	 	(\tau \in U_{\SUA}(\{m \stackrel{?}{=}_{\SUA} t\}))
	\end{array}
	 \right)
\]

From (\ref{e.tau-std-subs}) we have, $(\fa t \in \SubTerms(S_1))(t\subcomb = t\subiso)$, and hence we have that all the rules in  $\Rules \setminus \{\msf{un},\msf{ksub}\}$ are applicable on the target of the active constraint of  $\chisocs$,  if they were on  $\chcombcs$,  provided they are applied on a term in $\SubTerms(S_1)$:  

\begin{equation}\label{e.target-rules-apply-equally}
\begin{array}{l}
(\fa r \in \Rules) 
	\left(
		\begin{array}{c}
			\appl(r,\chcombcs,\chcombcs',\_,\_, =_{\SUA}) \wedge \\
			\act(m : \_,\chcombcs) \wedge  
			\act(m' : \_,\chcombcs') 	\wedge \\
			\act(m : \_,\chisocs)
		\end{array}
	\right)
	\Ra  \\
	\left(
		\begin{array}{c}
			\appl(r,\chisocs,\chisocs',\_,\_, =_{\SUA}) \wedge 
			\act(m' : \_,\chisocs')
		\end{array}
	\right).
\end{array}
\end{equation}

Similarly, all rules that are applicable on a term in the term set of the active constraint in  $\chcombcs$, say $c$,  are also applicable on the same term of the active constraint in  $\chisocs$, say $c'$ (provided the term exists in the term set of $c'$, which it does from (\ref{e.same-termsets}) and (\ref{e.rules-add-subterms})):

\begin{equation}\label{e.termset-rules-apply-equally}
\begin{array}{l}
(\fa r \in \Rules) 
	\left(
		\begin{array}{c}
			\appl(r,\chcombcs,\chcombcs',\_,\_, =_{\SUA}) \wedge \\
			\act(\_ : \_ \cup t,\chcombcs) \wedge 
			\act(\_ : \_ \cup T' ,\chcombcs') \wedge \\
			\act(\_ : \_ \cup t,\chisocs)
		\end{array}
	\right)
	\Ra \\
	\left(
		\begin{array}{c}
			\appl(r,\chisocs,\chisocs',\_,\_, =_{\SUA}) \wedge 
			\act(\_ : \_ \cup T',\chisocs')
		\end{array}
	\right).
\end{array}
\end{equation}

Finally, we can combine, (\ref{e.comb-satisfiable}),  (\ref{e.target-rules-apply-equally}), (\ref{e.termset-rules-apply-equally}), and (\ref{e.un-applies-equally}) to infer:

\begin{equation}\label{e.satisfiability-achieved-for-cs1}
	\left(
 \begin{array}{l}
	(\isocs = \lan \_ : \_, \ldots, \_ : T \ran) \wedge
  \appl(r_1,\isocs^{\frown}[\mi{sec} : T],\isocs_1,\{\}, \sigma_1, =_{\SUA}) \wedge \\
  \appl(r_2,\isocs'_1,\isocs_2,\sigma_1,\sigma_2, =_{\SUA}) \wedge \ldots \wedge \\
  \appl(r_p,\isocs'_{p-1},\isocs_p,\sigma_{p-1},\sigma_p,=_{\SUA}) \wedge \\
  \msf{simple}(\isocs_p) \wedge 
  (\fa i \in \{ 1,\ldots,p\})(\isocs'_i = \normalize(\isocs_i))
\end{array}
\right).
\end{equation}

\noindent
where $[r_1,\ldots,r_p]$ is a subsequence\footnote{$s'$ is a {\em subsequence} of a sequence $s$, if $s = \_^{\frown}s'^{\frown}\_$.} of $R$ (defined in \ref{e.comb-satisfiable}).

This in turn implies $\satisfiable(\isocs^{\frown}\mi{sec} : T,\sigma_p,=_{\SUA})$ from the definition of satisfiability.

We can then combine this with the fact that $S_1$ is a semi-bundle of $P_1$, and $\isocs$ is a constraint sequence of $S_1$ and conclude:

\[
	\begin{array}{c}
	  \semibundle(S_1,P_1) \wedge \conseq(\isocs,S_1) \wedge 
	  (\isocs = [\_ : \_, \ldots, \_ : T]) \wedge \\
	  \satisfiable(\isocs^{\frown}[\mi{sec} : T], \sigma_p,=_{\SUA}). 
	\end{array}
\]

But from Definition \ref{d.secrecy} ({\bf Secrecy}), this implies, $\neg \sfs(P_1,=_{\SUA})$, a contradiction to the hypothesis. Hence, $P_1$ is always secure for secrecy in the $=_{\SUA}$ theory, in combination with $P_2$ (or any other set of protocols) with which it is $\munutsat$.

\end{proof}

\section{Conclusion}\label{s.conclusion}

In this paper, we provided formal proofs that tagging to ensure non-unifiability of distinct encryptions prevents type-flaw and multi-protocol attacks under the {\sf ACUN} properties induced by the Exclusive-OR operator. We will now discuss some prospects for future work and related work.

\subsection{Future work}\label{ss.future-work}

Our results can be achieved under other equational theories the same way as we achieved them under the \tacun{} theory: When we use BSCA, the unification algorithms for the other theories will return an empty unifier, since their problems will have only constants as subterms. Hence, unifiers only from the standard unification algorithm need to be considered, which are always well-typed for $\nutsat{}$ protocols. In addition, this reasoning has to be given within a symbolic constraint solving model that takes the additional equational theories into account (the model we used, adapted from~\cite{Chev04} was tailored to accommodate only \tacun).

Our result on type-flaw attacks is obviously independent of security properties: It is valid for any property that can be tested on all possible protocol execution traces. Hence, we conjecture that it will also be valid for properties such as observational equivalence, which has been of interest to many protocol researchers of late (e.g.~\cite{BlanchetOakland06,DKR-csf08}). However, this property has been traditionally defined only in the applied pi-calculus. To use the results of this paper, we would have to first define an equivalent definition with symbolic constraint solving which is the model used in this paper (perhaps by extending~\cite{CMAE03}).

We achieved our result on multi-protocol attacks, specifically for secrecy. The reason for this was that, in order to prove that attacks exist in isolation, if they did in combination, we had to have a precise definition as to what an ``attack" was to begin with. However, other properties such as authentication and observational equivalence can be considered on a case-by-case basis with similar proof pattern. 

At the core of our proofs is the use of BSCA. However, their algorithm only works for disjoint theories that do not share any operators. For instance, the algorithm cannot consider equations of the form,

\[		[ a, b ]  \oplus  [ c, d ]   =   [ a \oplus c, b \oplus d ]. \]

We plan to expand our proofs to include such equations in future. However, it can be easily seen that the proof of Theorem~\ref{t.type-flaw} falls apart under this equation. For instance, consider the following unification problem:

\[	[\msf{nonce}, N_A] \stackrel{?}{=} [\msf{nonce},n_b] \oplus [\msf{agent}, a] \oplus [\msf{agent}, b].	\]

Now this problem is not unifiable under $=_{\SUA}$ theory, but it is when we add the new equation above to the theory, since $N_A$ can be substituted with $n_b \oplus a \oplus b$ to make the terms equal, which is an ill-typed substitution. 
It does not seem that a similar effect exists on multi-protocol attacks, but we intend to investigate further in that direction.

The most significant advantage of being able to prevent type-flaw attacks is that analysis could be restricted to well-typed runs only. This has been shown to assist decidability results in the standard, free theory~\cite{Lowe99,RS05} but not under monoidal theories. We are currently in a  pursuit to achieve a decidability result for protocol security in the presence of \txor.

\subsection{Related work}\label{ss.related-work}

To the best of our knowledge, the consideration of algebraic
properties and/or equational theories for type-flaw and multi-protocol
attacks is unchartered waters with the exception of a recent paper~\cite{CC10}.

\paragraph*{Type-flaw attacks.}
Type-flaw attacks on password protocols were studied by Malladi et
al. in~\cite{MA03}. That is the closest that we know about any study
of type-flaw attacks where the perfect encryption
assumption was relaxed. Some recent works studied type-flaw attacks using new
approaches such as rewriting~\cite{NG06}, and process calculus LySa
~\cite{GBD08}. However, they do not discuss type-flaw attacks under
operators with algebraic properties.

Recently in~\cite{ML09}, we gave a proof sketch that tagging prevents
type-flaw attacks even under \txor. The current paper is an extended,
journal version of~\cite{ML09} with the addition of a new result for
multi-protocol attacks.

A proof was presented in Malladi's PhD dissertation~\cite{Mall04} that type-flaw attacks can be prevented by component numbering with the constraint solving model of~\cite{MS01} as the framework. A similar proof approach was taken by  Arapinis et al.  in~\cite{AD07} using Comon et al.'s constraint solving model ~\cite{CCZ07} as the framework. In~\cite{CM07}, we used the proof style of~\cite{Mall04} to prove the decidability of tagged protocols that use \txor{} with the underlying framework of~\cite{Chev04} which extends~\cite{MS01} with \txor. That work is similar to our proofs since we too use the same framework (\cite{Chev04}). Further, we use BSCA as a core aspect of this paper along the lines of~\cite{CM07}.

\paragraph*{Multi-protocol attacks.} Kelsey et al. in their classical work~\cite{KSW97} showed that for any protocol, another protocol can be designed   to attack it. Cremers studied the feasibility of
multi-protocol attacks on published protocols and found many attacks,
thereby demonstrating that they are a genuine threat to protocol
security~\cite{Cremers06}. However, Cremers did not consider algebraic properties in the analysis.

A study of multi-protocol attacks with the perfect encryption 
assumption relaxed were first studied by Malladi et al. in~\cite{MAM02} 
through ``multi-protocol guessing attacks" on password protocols. 
Delaune et al. proved that these can be prevented by tagging 
in~\cite{DKR-csf08}.

The original work of Guttman et al. in
~\cite{TG00} assumed that protocols would not have type-flaw attacks when
they proved that tagging/disjoint encryption prevents multi-protocol
attacks. But a recent work by Guttman seems to relax that assumption
~\cite{Guttman09}. Both~\cite{TG00} and~\cite{Guttman09} use the strand space model~\cite{THG98}. Our protocol model in this paper is also based on strand spaces but the penetrator actions are modeled as symbolic reduction rules in the constraint solving algorithm of~\cite{Chev04,MS01}, as opposed to penetrator strands in~\cite{THG98}. 
Cortier-Delaune also seem to prove that multi-protocol attacks can be prevented with tagging, which is
slightly different from~\cite{TG00} and considers composed/non-atomic keys~\cite{CD09}. They too seem to use constraint satisfiability to model penetrator capabilities.

None of the above works considered the \txor{} operator or any
other operator that possesses algebraic properties.

In a recent paper that is about to appear in the CSF symposium, Ciobaca and Cortier seem to present protocol composition for arbitrary primitives under equational theories with and without the use of tagging~\cite{CC10}. Their results seem very general and broadly applicable. As future work, they comment in the conclusion of that paper that it is a challenging open problem to address  cases where multiple protocols uses \txor{}, which is solved in this paper.

\paragraph*{XOR operator.} Ryan and Schneider showed in~\cite{RS98} that new attacks can be launched on protocols when the algebraic properties of the \txor{} operator are exploited. In~\cite{CKRT03}, Chevalier et al. described the first NP-decision procedure to analyze protocols that use the \txor{} operator with a full consideration of its algebraic properties. We use an adapted version of their NSL protocol  in this paper as a running example. In an impressive piece of work, Chevalier also introduced a symbolic constraint solving algorithm for analyzing protocols with \txor{}, which we use as our framework in this paper~\cite{Chev04}. 

In an interesting work~\cite{KustersT08}, Kuesters and Truderung showed that the verification of protocols that use the \txor{} operator can be reduced to verification in a free term algebra, for a special class of protocols called $\oplus$-linear protocols\footnote{Kuesters-Truderung define a term to be $\oplus$-linear if for each of its subterms of the form $s \oplus t$, either $t$ or $s$ is ground.}, so that {\sf ProVerif} can be used for verification. 

Chen et al. recently report an extension of Kuesters-Truderung to improve the efficiency of verification by reducing the number of substitutions that need to be considered (thereby improving the performance of {\sf ProVerif}), and a new bounded process verification approach to verify protocols that do not satisfy the $\oplus$-linearity property~\cite{CDP09}.

These results have a similarity with ours, in the sense that we too show that the algebraic properties of \txor{} have no effect when some of the messages are modified to suit our requirements. 

A few months back, Chevalier-Rusinowitch report a nice way to compile cryptographic protocols into executable roles and retain the results for combination of equational theories in the context of compiling~\cite{CR09}. Like other works described above, their work does not seem to use tagging.

\paragraph*{Acknowledgments.} I benefited greatly from the following people's help and guidance: Jon Millen (MITRE, USA) clarified numerous concepts about constraint solving and some crucial aspects of \txor{} unification. Pascal Lafourcade (UFR IMA, France) gave several useful comments and reviews of the paper. Yannick Chevalier (IRIT, France) explained some concepts about his extensions to Millen-Shmatikov model with \txor. More importantly, our joint work toward decidability in~\cite{CM07} helped in laying the structure of proofs in this paper. 

\paragraph*{Funding.} This work funded in part by a doctoral SEED grant by the Graduate School at Dakota State University. I am particularly grateful to Dean~Tom~Halverson (college of BIS) and Dean~Omar~El-Gayar (college of graduate studies and research) for their continued support for my research.

\addcontentsline{toc}{section}{References}
\bibliographystyle{plain}
\bibliography{JSecRet-09}	

\newpage
\appendix
\section{Appendix}\label{s.appendix}

In the appendix, we first provide an index for the notation and terminology in Section~\ref{ss.index}. We then provide a detailed formalization of Baader \& Schulz Algorithm for combined theory unification~\cite{BS96} in Section~\ref{ss.BSCA}.

\subsection{Index - Notation and Terminology}\label{ss.index}


\subsubsection{Symbols}\label{sss.symbols}

\begin{tabular}{cll}
$[t_1,\ldots,t_n]$	&		&		Sequence of terms $t_1$ through 
$t_n$, that are linearly ordered.	\\
&		&		\\

$\penc{t}{k}$		&		&		$t$ encrypted using $k$ with an asymmetric encryption algorithm.	\\
&		&		\\

$\senc{t}{k}$		&		&		$t$ encrypted using $k$ with a symmetric encryption algorithm.		\\
&		&		\\

$h(t)$			&		&		The hash of $t$ using some hashing algorithm.	\\
&		&		\\

$\mi{sig}_k(t)$		&		&		The signature of $t$ using a private key that is verifiable with the \\ 
&		&		public-key $k$. \\
&		&		\\

$t_1 \oplus \ldots \oplus t_n$	&		&		Terms $t_1$ through $t_n$ XORed together.	\\
&		&		\\

$\frown$		&			&			$s_1^{\frown}s_2$ indicates  concatenation of two sequences $s_1$ and $s_2$.	\\
&		&		\\

$\prec_t$		&		&		A linear order relation obeyed by the 		elements of a sequence $t$; \\
&	&		Read $t_i \prec_t t_j$ as 
$t_i$ precedes $t_j$ in the sequence $t$; \\
&		&		\\

$\prod_{i = 1}^{n}  c_i$		&		&		Sequence concatenation of $c_1$ through $c_n$.   \\
&		&		\\

$\sqss$		&		&		Subterm relation; $t \sqss t'$ indicates 
$t$ is a part of $t'$. \\
&		&		\\

$\Subset$	&		&		Interm relation;  $t \Subset t'$ implies that $t$ equals $t'$ 									or an interm of one of the \\
					&		& 	elements of $t'$ if $t'$ is a sequence or is part of the 												plain-text, \\
					&		&		if $t'$ is an encryption;	\\
					&		&		\\

$\mc{P}(X)$	&		&		Power-set of $X$;	\\
&		&		\\

$x/X$		&		&		$x$ is substituted for the variable $X$;	\\
&		&		\\

$\sigma, \tau, \rho, \alpha, \beta$		&		&		Sets of substitutions;		\\
&		&		\\

$\Sigma$		&		&		Sets of sets of substitutions;	\\
&		&		\\

\end{tabular}

\begin{tabular}{cll}

$t =_{\Th} t'$	&		&		$t$ is equal to $t'$ in the theory $\Th$; \\
&		&		\\

$\Gamma$	&		&		A unification problem or a finite set of equations;		\\
&		&		\\

$A_{E}$			&		&		Unification algorithm that returns the most general unifiers \\ 
&		&		in the theory $=_E$ for a $E$-Unification Problem;	\\
&		&		\\

$(S \cup A)$	&		&		$\std \cup \acun$;		\\
&		&		\\

$\Gamma_3$		&		&		Obtained from $\Gamma_2$ such that, variables in $\Gamma_2$ are replaced by other \\
&		&	 variables in their equivalence classes in a variable identification \\ 
&		&		partition	on the variables called $\VarIdP$;	\\
&		&		\\

$\Gamma_{4.1}$	&		&		$\Gamma_3$	split into problems from only the theory $\Th_1$;	\\
&		&		\\

$\Gamma_{4.2}$	&		&		$\Gamma_3$	split into problems from only the theory $\Th_2$;	\\
&		&		\\

$V_1,V_2$		&		&		$\{ V_1, V_2 \}$ is a partition on the variables of $\Gamma_3$;		\\
&		&		\\

$\Gamma_{5.1}$	&		&		Variables in $\Gamma_{4.1}$ that belong to $V_2$ are replaced by new constants;	\\
&		&		\\

$\Gamma_{5.2}$	&		&		Variables in $\Gamma_{4.2}$ that belong to $V_1$ are replaced by new constants;	\\
&		&		\\

$\alpha, \beta$		&		&		The sets of substitutions for the replacement of variables with \\
&		&	 new constants in $\Gamma_{5.1}$ and $\Gamma_{5.2}$;	\\

$<$		&		&		$X < Y$ indicates that variable $X$ is not a subterm of an instantiation of $Y$;		\\
&		&		\\

$\odot$		&		&		$\sigma_1 \odot \sigma_2$ is the combined unifier of $\sigma_1$ and $\sigma_2$ in the theory $\Th_1 \cup \Th_2$,  \\
&		&	 if $\sigma_1$ is the unifier for $\Gamma_{5.1}$ and $\sigma_2$ is the unifier for $\Gamma_{5.2}$;	\\
&		&		\\

$m~:~T$		&		&		A constraint describing that term $m$ should be derivable by \\
&		&		using attacker actions on the set of terms $T$; \\
&		&		 \\

$+t$	&		&		A node that sends a term $t$; \\
&		&		\\

$-t$	&		&		A node that receives a term $t$; \\

\end{tabular}

\subsubsection{Words}\label{sss.words}

\begin{tabular}{cll}

&		&		\\

$\Vars$			&		&		Set of all variables;	\\
						&		&		\\

$\Constants$		&		&		Set of all constant values that are 
								indivisible (nonce, agent etc.)	\\
						&		&		\\

$\msf{in}$	&		&		$a ~\msf{in} ~as$ represents $a$ is an element in 
										the sequence $as$;	\\
						&		&		\\						

$T(F,\Vars)$		&		&		Term algebra; Set of all terms using function symbols \\
								&		&		$F$ and $\Vars$ 	\\			
								&		&		\\

$\mi{Terms}()$	&		&		Overloaded function returning all the 
										terms in a set of terms, strands, \\										
						&		&		or set of strands.			\\
						&		&		\\

$\type()$		&		&		Function returning the type of a term (agent, nonce, 	\\
						&		&		nonce encrypted with a public-key etc.)\\
						&		&		\\
											
$\SubTerms()$	&		&		Overloaded function returning all the 
										subterms in a set of terms, strands, \\								&		&		protocol, or semi-bundle.			\\
						&		&		\\															
						
$\EnCp()$		&		&		Overloaded function returning all the 
										encrypted subterms \\
						&		&	 	of a term, or set of strands; 	\\
						&		&		\\															
														
$\wellty()$	&		&		Predicate returning $\true$ if a 														substitution or sets of substitutions\\ 							
						&		&		are such that values are substituted to variables of the 												same type;		\\
						&		&		\\
						
$\std$			&		&		Set of identities involving $\StdOps$-Terms that is the basis \\
						&		&		for $=_{\std}$ theory; 	\\																								
					
$\acun$			&		&		Set of identities involving only the $\oplus$ operator to \\
						&		&		reflect it's \tacun{} algebraic properties 	\\													&		&		\\															
						
$\msf{disjoint}(\Th_1,\Th_2)$  &		&	 Predicate returning $\true$ if $\Th_1$ and $\Th_2$ do not share operators;	\\
						&		&		\\															

$\msf{ast}(t',t,\Th)$		&		&		Predicate returns $\true$ if $t'$ is a subterm of $t$ \\
						&		&	 	and made with operators not belonging to $\Th$; \\
						&		&		\\															
$\pure(t,\Th)$		&		&		Predicate returning $\true$ if $t$ has no alien subterms wrt \\
						&		&		operators of $\Th$;  \\
						&		&		\\																					

$\NewVars$	 &		&		A subset of $\Vars$ that did not 									previously appear in a \\
						&			&		unification problem; \\
						&		&		\\													
						
\end{tabular}

\begin{tabular}{cll}	

$\NewConstants$	 &		&		A subset of $\Constants$ that did 									not previously appear in a \\
							&			&		unification problem; \\
						&		&		\\															

$\VarIdP$		&		&		A partition on all the variables in a 																					unification problem;  \\
						&		&		\\

$\msf{least}(x,X,<)$		&		&		Returns $\true$	if $x$ is the minimal 										element of $X$ wrt the \\
									&		&		linear relation $<$;	\\					 							&		&		\\


$\mi{Node}$		&		&		Tuple $\tuple{\pm}{\Term}$  \\
						&		&		\\															
$\mi{Strand}$		&		&		Sequence of nodes  \\
						&		&		\\

$\FreshVars$		&		&	 Variables in a strand   that are of the \\
								&		&	 type nonce, session-key etc.; \\
								&		&		\\

$\LTKeys()$			&		&		Returns the set of subterms in a protocol that 
												resemble \\
 								&		&		$\sh(\_, \_)$;	\\
								&		&		\\
semi-strand	&		&	 Strand obtained by instantiating the known variables \\
								&		&		of a role;	\\
								&		&		\\
semi-bundle	&		&		Set of semi-strands;		\\
								&		&		\\

$\constraint(\tuple{m}{T})$ 	&		&	 $\true$ if $m~:~T$ is a constraint with $m$ as the target and \\
								&		&		$T$ as the termset \\
								&		&		\\

$\conseq(cs,S)$ 	&		&	 $cs$ is a constraint sequence from the semi-bundle $S$; \\
								&		&		\\
								
$\simple(c)$		&		&		$c$ is a constraint with only a variable on its target;  \\
								&		&		\\

$\simple(cs)$		&		&		$cs$ is a constraint sequence with only simple constraints;  \\
								&		&		\\

$\act(c,cs)$		&		&	 $\true$ if all constraints in $cs$, prior to $c$ are simple;  \\
								&		&		\\

$cs_<$			&		&	 Returns the constraint sequence prior to the active \\ 									&		&		constraint of $cs$;  \\
								&		&		\\

$cs_>$			&		&	 Returns the constraint sequence after to the active \\		  								&		&	 constraint of $cs$;  \\
								&		&		\\

\end{tabular}

\begin{tabular}{cll}

$\appl(r,cs,cs',\sigma,\sigma',\Th)$		&		&	 $\true$ if $r$ is applicable on $cs$, transforming it into $cs'$, and  \\
&		&		changing its substitution from $\sigma$ to $\sigma'$ in the theory $\Th$;  \\
							&		&		\\

$\normal(cs)$		&		&	 $\true$ if $cs$ has no free variables, or pairs  \\
			&		&		in the target or termset of its active constraint;  \\
								&		&		\\

$\normalize(cs)$		&		&	 Function that transforms $cs$ into a normal \\ 
&		&		constraint sequence	and returns it;  \\							&		&		\\

$\msf{typeFlawAttack}(P,\Th)$		&		&	 $\true$ if a constraint sequence from a semi-bundle of $P$ \\
	&		&		can only be	satisfied with an ill-typed substitution in the \\
	&		&		theory $\Th$;	\\
&		&		\\

\end{tabular}

\begin{tabular}{cll}

$\sfs(P,\Th)$		&		&	 $\true$ if protocol $P$ does not have a potential breach \\
	&		&		of secrecy in the theory $\Th$;		\\
&		&		\\

$\nutsat(P)$		&		&		$\true$ if $P$ satisfies three conditions including non-unifiable \\
&		&		encrypted subterms (in the $\SUA$ theory), no free \\
&		&		variables as asymmetric keys inside \\
&		&		\txor{} terms;  \\
&		&		\\

$\munutsat(P_1,P_2)$		&		&	 $\true$ if	encrypted subterms of $P_1$ are non-unifiable with the \\
&		&		 encrypted subterms of $P_2$, in the $\SUA$ theory;  \\
&		&		\\
								
\end{tabular}

\subsection{Bader \& Schulz Combined Theory Unification Algorithm (BSCA)}\label{ss.BSCA}

We will now consider how two UAs for two disjoint theories $=_{E_1}$ and $=_{E_2}$, may be combined to output the unifiers for $(E_1 \cup E_2)$-UPs using Baader \& Schulz Combination Algorithm (BSCA)~\cite{BS96}.

We will use the following $(\SUA)$-UP as our running example\footnote{We omit the superscript $\to$ on encrypted terms in this problem, since they obviously use only asymmetric encryption.}:

\[ \left\{ [1,n_a]_{\mi{pk}(B)} \stackrel{?}{=}_{\SUA} [1,N_B]_{\mi{pk}(a)}   \oplus   [2,A] \oplus  [2,b] \right\}. \]

BSCA takes as input a $(E_1 \cup E_2)$-UP, say $\Gamma$, and applies some transformations on them to derive $\Gamma_{5.1}$ and $\Gamma_{5.2}$ that are $E_1$-UP and $E_2$-UP respectively. 

\subsubsection*{Step 1 (Purify terms)} BSCA first ``purifies" the given set of $(E = E_1 \cup E_2)$-UP, $\Gamma$, into a new set of problems $\Gamma_1$, such that, all the terms are pure wrt $=_{E_1}$ or $=_{E_2}$.

If our running example was $\Gamma$, then, the set of problems in $\Gamma_1$ are $W \stackrel{?}{=}_{\std} [1,n_a]_{\mi{pk}(B)}$, $X  \stackrel{?}{=}_{\std}  [1,N_B]_{\mi{pk}(a)}, Y  \stackrel{?}{=}_{\std} [2,A]$, $Z  \stackrel{?}{=}_{\std}  [2,b]$, and $W \stackrel{?}{=}_{\acun} X \oplus Y \oplus Z$, where $W, X, Y, Z$ are obviously new variables that did not exist in $\Gamma$.

\subsubsection*{Step 2. (Purify problems)} Next, BSCA purifies $\Gamma_1$ into $\Gamma_2$ such that, every problem in $\Gamma_2$ has both terms pure wrt the same theory.

For our example problem, this step can be skipped since all the problems in $\Gamma_1$ already have both their terms purely from the same theory ($=_{\std}$ or $=_{\acun}$)). 

\subsubsection*{Step 3. (Variable identification)} Next, BSCA partitions $\Vars(\Gamma_2)$ into a partition $\VarIdP$ such that, each variable in $\Gamma_2$ is replaced with a representative from the same equivalence class in $\VarIdP$. The result is $\Gamma_3$.

In our example problem, one set of values for $\VarIdP$ can be 
\[ \left\{ \{A\},\{B\},\{N_B\}, \{W\},\{X\},\{Y,Z\} \right\}. \]

\subsubsection*{Step 4. (Split the problem)} The next step of BSCA is to split $\Gamma_3$ into two UPs $\Gamma_{4.1}$ and $\Gamma_{4.2}$ such that, each of them has every problem with terms from the same theory, $\Th_1$ or $\Th_2$.
 
Following this in our example,

\[  \Gamma_{4.1} = \left\{ W \stackrel{?}{=}_{\std} [1,n_a]_{\mi{pk}(B)}, X  \stackrel{?}{=}_{\std} [1,N_B]_{\mi{pk}(a)}, Y \stackrel{?}{=}_{\std} [2,A], Z \stackrel{?}{=}_{\std} [2,b]  \right\}, \] and 

\[ \Gamma_{4.2} = \left\{ W \stackrel{?}{=}_{\acun} \xorseq{X}{Y}{Y} \right\}. \]

\subsubsection*{Step 5. (Solve systems)} The penultimate step of BSCA is to partition all the variables in $\Gamma_3$ into a size of two: Let $p = \{ V_1, V_2 \}$ is a partition of $\Vars(\Gamma_3)$. Then, the earlier problems ($\Gamma_{4.1}$, $\Gamma_{4.2}$) are further split such that, all the variables in one set of the partition are replaced with new constants in the other set and  vice-versa. The resulting sets are $\Gamma_{5.1}$ and $\Gamma_{5.2}$.

In our sample problem, we can form $\{ V_1, V_2 \}$ as $\{ \Vars(\Gamma_3), \{\} \}$. i.e., we choose that all the variables in problems of $\Gamma_{5.2}$ be replaced with new constants. This is required to find the unifier for the problem (this is the partition that will successfully find a unifier).

So $\Gamma_{5.1}$ stays the same as $\Gamma_{4.1}$, but $\Gamma_{5.2}$ is changed to 

\[	\Gamma_{5.2} = \Gamma_{4.2} \beta  
=	 \left\{ W \stackrel{?}{=}_{\acun} \xorseq{X}{Y}{Y} \right\} \beta  =  \left\{ w \stackrel{?}{=}_{\acun} \xorseq{x}{y}{y} \right\}.
\]

 i.e., $\beta = \left\{ w/W, x/X, y/Y \right\}$, where, $w, x, y$ are constants, which obviously did not appear in $\Gamma_{5.1}$. 

\subsubsection*{Step 6. (Combine unifiers)} The final step of BSCA is to combine the unifiers for $\Gamma_{5.1}$ and $\Gamma_{5.2}$, obtained using $A_{E_1}$ and $A_{E_2}$. This was given in Def.~\ref{d.Combined-Unifier}.

\end{document}